\documentclass[12pt]{elsarticle}

\usepackage[T1]{fontenc}
\usepackage{graphicx} 
\usepackage{algpseudocode}
\usepackage{tikz}
\usetikzlibrary{automata}

\usepackage{subcaption}

\usepackage[section]{placeins} 
\usepackage{graphicx} 

\usepackage[section]{placeins} 

\usepackage[algoruled,linesnumbered,titlenotnumbered,noend]{algorithm2e}

\SetKwBlock{Let}{let}{}
\SetCommentSty{mycommfont}
\DontPrintSemicolon 
\usepackage{amssymb}
\usepackage{pifont,bbding,amsmath,stmaryrd}
\usepackage{amsthm}

\newtheorem{theorem}{Theorem}[section]
\newtheorem{definition}{Definition}
\newtheorem{corollary}[theorem]{Corollary}
\newtheorem{lemma}[theorem]{Lemma}
\newtheorem{proposition}[theorem]{Proposition}
\newtheorem{example}[theorem]{Example}
\newtheorem{remark}[theorem]{Remark}

\usepackage{xcolor}
\usepackage{pgf}
\usepackage{pgfmath}
\usepackage{wrapfig} 
\usepackage{xspace} 
\usepackage[capitalise]{cleveref}
\usepackage{tikz}
\usetikzlibrary{calc,positioning,backgrounds,decorations.pathmorphing,shapes,calc,arrows.meta}
\usepackage{tikz-layers}
\usetikzlibrary{automata}
\usetikzlibrary{shapes.multipart}
\usepackage{mathtools}
\usepackage{complexity}
\usepackage{comment}

\usepackage{caption}
\usepackage{subcaption}
\usepackage{pgfplots}
\usepackage{nicefrac,xfrac}

\usepackage{fancyhdr}
\let\oldnl\nl
\newcommand{\nonl}{\renewcommand{\nl}{\let\nl\oldnl}}

\tikzstyle{irnode}=[font=\Large]
\tikzstyle{bnfnode}=[font=\Large]
\colorlet{eventcolor}{green!50!black}
\colorlet{msgcolor}{red!70!black}
\colorlet{typecolor}{blue!80!black}
\colorlet{keycolor}{green!50!black}
\colorlet{valcolor}{red}
\colorlet{attrcolor}{blue}
\colorlet{opcolor}{eventcolor}
\colorlet{stmtcolor}{eventcolor}
\colorlet{procedurenamecolor}{green!50!black}
\colorlet{relcolor}{blue}
\definecolor{varcolor}{RGB}{166,42,42} 
\definecolor{proccolor}{rgb}{0.82, 0.1, 0.26} 
\definecolor{threadcolor}{rgb}{0.82, 0.1, 0.26} 

\newcommand\true{\mathtt{true}}
\newcommand\false{\mathtt{false}}
\newcommand{\tuple}[1]{\left\langle#1\right\rangle}

\newcommand\of[2]{{#1}\left(#2\right)}

\newcommand\tof[4]{{#1}\left(#2\right)\left(#3\right)\left(#4\right)}

\newcommand\ii{i}

\newcommand\nn{n}

\newcommand\rel{R}
\newcommand\rclosureof[1]{{#1}^?}
\newcommand\tclosureof[1]{{#1}^+}

\newcommand\mkrelof[1]{\,\textcolor{relcolor}{\left[#1\right]}\,}
\newcommand\invmkrelof[1]{{#1}^{-1}}
\newcommand\fundom{{\mathtt{dom}}}
\newcommand\fundomof[1]{\of\fundom{#1}}

\newcommand\fun{f}

\newcommand\compose\circ

\newcommand\undf\bot
\newcommand\funtype[3]{{#1}:{#2}\rightarrow{#3}}

\newcommand\updatefun[3]{{#1}\left[#2\rightarrow#3\right]}

\newcommand\xmodels[1]{\models^{#1}}

\newcommand\set[1]{\left\{{#1}\right\}}
\newcommand\setcomp[2]{\left\{{#1}\mid\,{#2}\right\}}
\newcommand\addtoset\oplus

\newcommand\assigned{:=}
\newcommand\ordering\preceq

\newcommand\emptyword\epsilon

\newcommand\wtranspose\otriangle

\newcommand\app\bullet

\newcommand\wfilter\odot

\newcommand\cone{\mbox{\ding{172}}}
\newcommand\ctwo{\mbox{\ding{173}}}
\newcommand\cthree{\mbox{\ding{174}}}
\newcommand\cfour{\mbox{\ding{175}}}
\newcommand\cfive{\mbox{\ding{176}}}
\newcommand\csix{\mbox{\ding{177}}}

\newcommand\irspace{\;\;\;\;\;}

\newcommand\valset{\mathtt{\mathcal{D}}}
\newcommand\val{\mathtt{v}}

\newcommand\zero{0}

\newcommand\varset{{\mathcal{V}}}

\newcommand\regset{\mathtt{Regs}}

\newcommand\xvar{x}
\newcommand\yvar{y}
\newcommand\avar{a}

\newcommand\creg{c}

\newcommand\areg{a}

\newcommand\breg{b}

\newcommand\thread\theta
\newcommand\pthread\phi

\newcommand\threadset{\Theta}

\newcommand\movesto[1]{\xrightarrow{#1}}

\newcommand\run\rho

\newcommand\runset{{\mathtt {Runs}}}
\newcommand\runsetof[1]{\of\runset{#1}}

\newcommand\conf\gamma

\newcommand\iconf[1]{\conf_{#1}}
\newcommand\confset{\Gamma}
\newcommand\initconf{\conf_{\it init}}
\newcommand\qstate{q} 
\newcommand\mystate{q} 
\newcommand\pqstate{\qstate'}

\newcommand\istate[1]{\mystate_{#1}}
\newcommand\pstate{\mystate'}

\newcommand\stateset{Q}
\newcommand\initstate{\mystate_{\mathtt{init}}}
\newcommand\transitionset\Delta
\newcommand\op{\textcolor{opcolor}{\mathtt{o}}}

\newcommand\regvalmapping{{\mathcal R}}
\newcommand\regvalmappingof[1]{\of\regvalmapping{#1}}
\newcommand\pregvalmapping{\regvalmapping'}

\newcommand\event{{\color{eventcolor}{e}}}
\newcommand\pevent{{\color{eventcolor}{\event'}}}
\newcommand\ppevent{{\color{eventcolor}{\event''}}}
\newcommand\revent{{\color{eventcolor}{\ensuremath{\mathsf{r}}}}}
\newcommand\wevent{{\color{eventcolor}{\ensuremath{\mathsf{w}}}}}
\newcommand\xinitevent[1]{{\color{eventcolor}{init^{#1}}}}
\newcommand\events{E}
\newcommand\pevents{\events'}
\newcommand\initeventset{{\mathtt{InitEvents}}}

\newcommand\ievent[1]{{\color{eventcolor}{\event_{#1}}}}
\newcommand\eventset{\mathtt{Events}}
\newcommand\silentevent{{\color{eventcolor}\tau}}

\newcommand\varattr{{\color{attrcolor}{\mathtt{var}}}}
\newcommand\varattrof[1]{{#1}\color{attrcolor}{{\cdot}\varattr}}
\newcommand\valattr{{\color{attrcolor}{\mathtt{val}}}}
\newcommand\valattrof[1]{{#1}{\color{attrcolor}{{\cdot}\valattr}}}
\newcommand\threadattr{{\color{attrcolor}{\mathtt{thread}}}}
\newcommand\threadattrof[1]{{#1}{\color{attrcolor}{{\cdot}\threadattr}}}
\newcommand\typeattr{{\color{attrcolor}{\mathtt{type}}}}
\newcommand\typeattrof[1]{{#1}{\color{attrcolor}{{\cdot}\typeattr}}}

\newcommand\egraph{G}
\newcommand\porel{{\color{relcolor}{\mathtt{{po}}}}}
\newcommand\xporel[1]{{\color{relcolor}{\porel_{#1}}}}

\newcommand\rfrel{{\color{relcolor}{\mathtt{{rf}}}}}
\newcommand\rfrelsof[1]{{\color{relcolor}{\mathtt{{RF}}}}}

\newcommand\corel{{\color{relcolor}{\mathtt{{co}}}}}

\newcommand\corelsof[1]{{\color{relcolor}{\mathtt{{CO}}}}}
\newcommand\xcorel[1]{{\color{relcolor}{\mathtt{{\corel_{#1}}}}}}

\newcommand\partcorel{{\color{relcolor}{\mathtt{{pco}}}}}

\newcommand\hbrel{{\color{relcolor}{\mathtt{{hb}}}}}

\newcommand\node{{\mathtt n}}

\newcommand\rasem{{\ensuremath{\mathtt{RA}}}\xspace}
\newcommand\wrasem{{\ensuremath{\mathtt{WRA}}}\xspace}
\newcommand\srasem{{\ensuremath{\mathtt{SRA}}}\xspace}

\newcommand\psisem{{\ensuremath{\mathtt{PSI}}}\xspace}
\newcommand\scsem{{\ensuremath{\mathtt{SC}}}\xspace}

\newcommand\bluecircled[1]{\mathop{\tikz{\node[draw,circle,fill=blue!15, inner sep=0.5pt,font={\footnotesize}]{#1};}}}

\newcommand{\bigO}[1]{\ensuremath{\mathcal{O}(#1)}}

\newcommand\transcup\Cup

\newcommand\reldelete\ominus

\newcommand\largebluecircled[1]{\mathop{\tikz{\node[draw,circle,fill=blue!15, inner sep=1pt]{#1};}}}

\newcommand\sthread\psi
\newcommand\ethread\eta

\newcommand\lbl\lambda

\newcommand\lblset\Lambda

\newcommand\mmovesto[2][]{\ext@arrow 0359\Rightarrowfill@{#1}{#2}}
\makeatother

\newcommand\pth\pi
\newcommand\ipth[1]{\pth_{#1}}

\newcommand\iqstate[1]{\qstate_{#1}}

\newcommand\opseq\rho
\newcommand\emptyopseq\epsilon

\newcommand\pwevent{{\color{eventcolor}{\wevent'}}}

\newcommand\initegraph{\egraph_{\mathsf{init}}}

\newcommand\mkegraph{{\mathsf{mkEgraph}}}
\newcommand\mkegraphof[1]{\of\mkegraph{#1}}

\newcommand\frrel{{\color{relcolor}{\mathsf{{fr}}}}}

\newcommand\frrelsof[1]{{\color{relcolor}{\mathsf{{FR}}}}}

\newcommand\concurrentwith\parallel

\tikzstyle{irnamenode}=
      [text=blue!70!black,anchor=west,align=left,
      font=\sffamily\Large\bfseries,anchor=west]

\newcommand\threeinfer[4]{
\begin{tikzpicture}
  \node[irnode,name=n1]{
    ${#1}$};
  \node[irnode,name=n2,anchor=north] at ($(n1.south)+(0pt,-3pt)$){
    ${#2}$};
  \node[irnode,name=n3,anchor=north] at ($(n2.south)+(0pt,-3pt)$){
    ${#3}$};
  \node[irnode,name=n4,anchor=north] at ($(n3.south)+(0pt,-6pt)$){
     ${#4}$};
  
     \path let \p1=(n1.west),\p2=(n2.west),\p3=(n3.south west),\p4=(n4.north west) in 
     coordinate (Q1) at ({min(\x1,\x2,\x3,\x4)}, {(\y3+\y4)/2}){};

     \path let \p1=(n1.east),\p2=(n2.east),\p3=(n3.south east),\p4=(n4.north east) in 
     coordinate (Q2) at ({max(\x1,\x2,\x3,\x4)}, {(\y3+\y4)/2});
     \draw[-,line width=1pt] (Q1) -- (Q2);
   \end{tikzpicture}
}

\newcommand\onelabelinfer[3]{
  \begin{tikzpicture}[framed, background rectangle/.style={draw=black,fill=black!5,rounded corners=1ex}]
    \node[irnode,name=n1]{${#1}$};
    \node[irnode,name=n2,anchor=north] at ($(n1.south)+(0pt,-6pt)$){${#2}$};
    \path let \p1=(n1.south west),\p2=(n2.north west) in 
    coordinate (Q1) at ({min(\x1,\x2)}, {(\y1+\y2)/2}){};
    \path let \p1=(n1.south east),\p2=(n2.north east) in 
    coordinate (Q2) at ({max(\x1,\x2)}, {(\y1+\y2)/2});
    \draw[-,line width=1pt] (Q1) -- (Q2);
     \path let \p1=(n1.east),\p2=(n2.east),
     \p3=(n1.north),\p4=(n2.south) in 
     coordinate (Q3) at ({max(\x1,\x2)}, {(\y3+\y4)/2});
     \node[irnamenode] at ($(Q3.east)+(0pt,0pt)$){#3};
  \end{tikzpicture}
}
\newcommand\twolabelinfer[4]{
  \begin{tikzpicture}[framed, background rectangle/.style={draw=black,fill=black!5,rounded corners=1ex}]
    \node[irnode,name=n1]{${#1}$};
    \node[irnode,name=n2,anchor=north] at ($(n1.south)+(0pt,-3pt)$){${#2}$};
    \node[irnode,name=n3,anchor=north] at ($(n2.south)+(0pt,-6pt)$){${#3}$};
    \path let \p1=(n1.west),\p2=(n2.south west),\p3=(n3.north west) in 
    coordinate (Q1) at ({min(\x1,\x2,\x3)}, {(\y2+\y3)/2}){};
    \path let \p1=(n1.east),\p2=(n2.south east),\p3=(n3.north east) in 
    coordinate (Q2) at ({max(\x1,\x2,\x3)}, {(\y2+\y3)/2});
    \draw[-,line width=1pt] (Q1) -- (Q2);
     \path let \p1=(n1.east),\p2=(n2.east),\p3=(n3.east),
     \p4=(n1.north),\p5=(n3.south) in 
     coordinate (Q3) at ({max(\x1,\x2,\x3)}, {(\y4+\y5)/2});
     \node[irnamenode] at ($(Q3.east)+(0pt,0pt)$){#4};
  \end{tikzpicture}
}
\newcommand\threelabelinfer[5]{
\begin{tikzpicture}[framed, background rectangle/.style={draw=black,fill=black!5,rounded corners=1ex}]
  \node[irnode,name=n1]{
    ${#1}$};
  \node[irnode,name=n2,anchor=north] at ($(n1.south)+(0pt,-3pt)$){
    ${#2}$};
  \node[irnode,name=n3,anchor=north] at ($(n2.south)+(0pt,-3pt)$){
    ${#3}$};
  \node[irnode,name=n4,anchor=north] at ($(n3.south)+(0pt,-6pt)$){
     ${#4}$};
     \path let \p1=(n1.west),\p2=(n2.west),\p3=(n3.south west),\p4=(n4.north west) in 
     coordinate (Q1) at ({min(\x1,\x2,\x3,\x4)}, {(\y3+\y4)/2}){};
     \path let \p1=(n1.east),\p2=(n2.east),\p3=(n3.south east),\p4=(n4.north east) in 
     coordinate (Q2) at ({max(\x1,\x2,\x3,\x4)}, {(\y3+\y4)/2});
     \draw[-,line width=1pt] (Q1) -- (Q2);
     \path let \p1=(n1.east),\p2=(n2.east),\p3=(n3.east),\p4=(n4.east),
     \p5=(n1.north),\p6=(n4.south) in 
     coordinate (Q3) at ({max(\x1,\x2,\x3,\x4)}, {(\y5+\y6)/2});
     \node[irnamenode] at ($(Q3.east)+(0pt,0pt)$){#5};
\end{tikzpicture}
}

\newcommand\fivelabelinfer[7]{
\begin{tikzpicture}[framed, background rectangle/.style={draw=black,fill=black!5,rounded corners=1ex}]
  \node[irnode,name=n1]{
    ${#1}$};
  \node[irnode,name=n2,anchor=north] at ($(n1.south)+(0pt,-3pt)$){
    ${#2}$};
  \node[irnode,name=n3,anchor=north] at ($(n2.south)+(0pt,-3pt)$){
    ${#3}$};
  \node[irnode,name=n4,anchor=north] at ($(n3.south)+(0pt,-3pt)$){
    ${#4}$};
  \node[irnode,name=n5,anchor=north] at ($(n4.south)+(0pt,-3pt)$){
     ${#5}$};
  \node[irnode,name=n6,anchor=north] at ($(n5.south)+(0pt,-6pt)$){
     ${#6}$};
     \path let \p1=(n1.west),\p2=(n2.west),\p3=(n3.west),
     \p4=(n4.west),\p5=(n5.south west),\p6=(n6.north west) in 
     coordinate (Q1) at ({min(\x1,\x2,\x3,\x4,\x5,\x6)}, {(\y5+\y6)/2}){};
     \path let \p1=(n1.east),\p2=(n2.east),\p3=(n3.east),
     \p4=(n4.east),\p5=(n5.south east),\p6=(n6.north east) in 
     coordinate (Q2) at ({max(\x1,\x2,\x3,\x4,\x5,\x6)}, {(\y5+\y6)/2});
     \draw[-,line width=1pt] (Q1) -- (Q2);
    \path let
     \p1=(n1.east),\p2=(n2.east),\p3=(n3.east),\p4=(n4.east),\p5=(n5.east),
     \p6=(n6.east),
     \p7=(n1.north),\p8=(n6.south) in 
     coordinate (Q3) at ({max(\x1,\x2,\x3,\x4,\x5)}, {(\y6+\y7)/2});
     \node[irnamenode] at ($(Q3.east)+(0pt,0pt)$){#7};
\end{tikzpicture}
}

\newcommand\rmachine{{\mathcal M}}
\newcommand\irmachine[1]{\rmachine_{#1}}
\newcommand\rmachinetuple{\tuple{{\stateset,\initstate,\transitionset}}}

\newcommand\inopof[4]{\textcolor{opcolor}{\left(\wtype,#1,#2,#3,#4\right)}}

\newcommand\outopof[3]{\textcolor{opcolor}{\left(\rtype,#1,#2,#3\right)}}
\newcommand\opof[4]{\textcolor{opcolor}{\left(#1,#2,#3,#4\right)}}

\newcommand\wopof[3]{\textcolor{opcolor}{\left(\wtype,#1,#2,#3\right)}}

\newcommand\ropof[3]{\textcolor{opcolor}{\left(\rtype,#1,#2,#3\right)}}

\newcommand\action{\textcolor{opcolor}{\alpha}}

\newcommand\ttype{\textcolor{eventcolor}{\mathsf{ty}}}
\newcommand\wtype{\textcolor{eventcolor}{\mathsf{W}}}
\newcommand\rtype{\textcolor{eventcolor}{\mathsf{R}}}

\newcommand\eventof[4]{\textcolor{eventcolor}{\left(#1,#2,#3,#4\right)}}
\newcommand\reventof[3]{\textcolor{eventcolor}{\left(\rtype,#1,#2,#3\right)}}
\newcommand\weventof[3]{\textcolor{eventcolor}{\left(\wtype,#1,#2,#3\right)}}
\newcommand\ceventof[2]{\textcolor{eventcolor}{#1 \assigned #2}}

\newcommand\difftype{{\mathtt{Diff}}}

\newcommand\raconsprob{\operatorname{\mathtt{RA-Cons}}}
\newcommand\sraconsprob{\operatorname{\mathtt{SRA-Cons}}}
\newcommand\wraconsprob{\operatorname{\mathtt{WRA-Cons}}}

\newcommand\addevent\oplus
\newcommand\deletenode\ominus
\newcommand\keepnode\odot

\newcommand\rmovesto[2]{\xrightarrow{#2}_{#1}}
\newcommand\cmovesto[1]{\xRightarrow{#1}}

\newcommand\failconf{{\tt \textcolor{relcolor}{Unsafe}}}

\newcommand\visible{{\textcolor{relcolor}{\mathsf{fragile}}}}
\newcommand\visibleof[5]{{#1}\mkrelof{\tof\visible{#2}{#3}{#4}}{#5}}
\newcommand\hidden{{\textcolor{relcolor}{\mathsf{exposed}}}}
\newcommand\hiddenof[5]{{#1}\mkrelof{\tof\hidden{#2}{#3}{#4}}{#5}}

\newcommand\status{{\textcolor{relcolor}{\mathsf{status}}}}
\newcommand\statusof[5]{{#1}\mkrelof{\tof\status{#2}{#3}{#4}}{#5}}

\newcommand\regs{{\textcolor{relcolor}{\mathsf{regs}}}}

\newcommand\pocolor{black}

\tikzstyle{poedge}=[black,->,line width=1pt]
\tikzstyle{rfedge}=[green!50!black,->,dashed,line width=1pt]
\tikzstyle{coedge}=[red!70!black,->,dotted,line width=1pt]

\newcommand\anysem{{\ensuremath{\mathtt{CM}}}\xspace}

\newcommand\eid\iota

\newcommand\morel{{\color{relcolor}{\mathtt{{mo}}}}}
\newcommand\xmorel[1]{{\color{relcolor}{\morel_{#1}}}}

\tikzstyle{irnamenode}=
      [text=blue!70!black,anchor=west,align=left,
      font=\sffamily\Large\bfseries,anchor=west]

\journal{VSI: ICTAC 2025}

\begin{document}

\begin{frontmatter}

\title{The complexity of verifying the release-acquire semantics over register machines\footnote{The authors were supported by Grant VR 2020-04430 of the Swedish Research Council. Elli Anastasiadi’s work is funded by the Villum Investigator Grant S4OS of the
 Danish Independent Research Fund. We would also like to 
 thank Nathan Lhote for his help with the \PSPACE~hardness
  proof of Section~\ref{sec:hardness}.}} 

\author[UU]{Parosh~Aziz~Abdulla}

\author[AAU]{Elli~Anastasiadi} 
\author[UU]{Mohamed~Faouzi~Atig}
\author[eiffel]{L\'{e}o~Exibard}
\author[UU]{Samuel~Grahn}

\affiliation[UU]{organization={Department of Information Technology, Uppsala University},
             city={Uppsala},
             country={Sweden}}
\affiliation[AAU]{organization={Department of Computer Science, Aalborg University},
             city={Aalborg},
             country={Denmark}}
\affiliation[eiffel]{organization={LIGM, CNRS, Univ Gustave Eiffel},
             city={F77454 Marne-la-Vallée},
             country={France}}

\begin{abstract}
The Release-Acquire (\rasem) semantics and its variants are some of the most fundamental models of concurrent semantics for architectures, programming languages, and distributed systems.
  Several steps have been taken in the direction of \textit{testing} such semantics, where one is interested in whether a single program execution is consistent with a memory model.
  The more general \textit{verification} problem, i.e., checking whether
  any allowed program run is consistent with a memory model,
   has still not been studied as much.
  The purpose of this work is to bridge this gap.
  We tackle the verification problem,
   where, given an  implementation described as a register machine,
   we check if any of its runs violates the $\rasem$ semantics or
   its Strong ($\srasem$) and Weak ($\wrasem$) variants.
  We show that verifying $\wrasem$ in this setup is in $\bigO{n^5}$,
   while verifying the $\rasem$ and $\srasem$ is \PSPACE~complete. This both answers some fundamental
    questions about the complexity of these problems, but also provides
    insights on the expressive power of register machines as a model.

\end{abstract}


\end{frontmatter}



\section{Introduction}\label{sec:intro}

Over the years, numerous consistency models have been proposed to capture the subtle concurrency semantics of hardware architectures, programming languages, and distributed systems.
The Release-Acquire ($\rasem$) semantics and its variants are some of the most fundamental consistency models weaker than sequential
consistency, which are especially common and well-studied in programming languages and distributed data stores.
Such consistency models allow different processes (threads) to have different views of the order of certain memory updates and maintain a looser global consensus on all events.
This allows for much faster implementations while still providing the user an intuitive and deterministic understanding of the underlying concurrency model.

$\rasem$ is a fragment of the C11 model \cite{DBLP:conf/popl/LahavGV16}, obtained by restricting the threads' write and read instructions to be release and acquire accesses, respectively.
The $\rasem$ model is appropriate as a rigorous foundational semantics on its own, independently of particular architectures and compilers, and it has verified compilation schemes to popular platforms such as the x86-TSO, POWER, and ARM architectures \cite{DBLP:conf/popl/BattyMOSS12,DBLP:conf/popl/BattyOSSW11,DBLP:conf/pldi/SarkarMOBSMAW12}.
Several variants of the $\rasem$ semantics have been proposed in the literature in recent years.
Notably, the Strong-Release-Acquires ($\srasem$) semantics \cite{DBLP:conf/popl/LahavGV16} strengthens $\rasem$ by forbidding behaviors that require the re-ordering of write instructions but coincides with $\rasem$ for programs that do not contain write-write races.
In \cite{DBLP:conf/popl/LahavGV16}, it is shown that \srasem captures precisely the guarantees provided by POWER compilers for programs compiled from $\rasem$.
Another variant is the Weak-Release-Acquire (\wrasem) semantics that has been considered as an alternative for \rasem in the semantics of shared-memory concurrent programs, permitting more efficient verification frameworks such as stateless model checking \cite{DBLP:journals/pacmpl/Kokologiannakis18}.

The relevance of \rasem and its siblings goes beyond compilers and hardware architectures.
At the distributed systems level, they are equivalent to standard and well-studied variants of causal consistency \cite{DBLP:journals/toplas/LahavB22}.
\srasem  corresponds to the causal convergence consistency semantics implemented in  data stores \cite{DBLP:conf/popl/BouajjaniEGH17,DBLP:journals/ftpl/Burckhardt14}
, while \wrasem corresponds to the classical definition of causal consistency \cite{DBLP:conf/popl/BouajjaniEGH17}.

One of the most fundamental computational problems for a given consistency model $\anysem$ is {\it consistency checking}.
Consistency checking comes in two flavors: {\it testing} and {\it verification} \cite{DBLP:journals/pacmpl/TuncA0K0P23,DBLP:journals/iandc/AlurMP00,DBLP:journals/siamcomp/GibbonsK97,DBLP:journals/pacmpl/AbdullaAJN18,DBLP:conf/asplos/LuoD21}.
In testing, we are given the consistency model $\anysem$, often described using a set of axioms, and a program run $\run$ consisting of a sequence of events.
The sequence is typically generated by an implementation, e.g., a hardware architecture, a compiler, or a distributed protocol,  that is supposed to guarantee $\anysem$.
The task is to check whether $\run$ satisfies $\anysem$.
The verification problem is more general: we are given an {\it implementation} and asked to check whether {\it all} executions of the implementation satisfy $\anysem$.

%
%
%
%

The relevance and intricacy of the \rasem-like semantics have led to several recent works checking their consistency.
All these works consider the {\it testing} problem.
The first results showed that testing consistency under the \rasem semantics is of polynomial complexity \cite{DBLP:journals/pacmpl/AbdullaAJN18,DBLP:journals/pacmpl/KokologiannakisLV23}.
Recently, it was shown that testing consistency checking for \srasem and \wrasem also have polynomial complexity \cite{DBLP:journals/pacmpl/TuncA0K0P23}.
Despite the above results on testing consistency, little is known about the complexity of \textit{verifying} consistency under \rasem semantics.
As far as we know, the problem is still poorly understood.
The goal of this work is to bridge this gap.

\textbf{Related Work}\label{sec:related_work}.
 In their seminal work \cite{DBLP:journals/siamcomp/GibbonsK97},
  Korach and Gibbons showed consistency testing under the SC semantics
   is \NP-hard.
Alur et al. showed that the
 verification problem under the SC
 semantics \cite{DBLP:journals/iandc/AlurMP00} is undecidable,
  albeit for a data-dependent implementation model.
   Hence, one that registers machines cannot capture.

Several examples of protocols (\cite{DBLP:journals/dc/AhamadNBKH95,DBLP:conf/ppopp/PerrinMJ16}) are designed to enforce different consistency models.
Such works guarantee the designed mechanisms' correctness and provide a good baseline for implementing practices.
However, they do not produce uniform frameworks 
 to answer the general consistency-checking question.

Bouajjani et al. \cite{DBLP:conf/popl/BouajjaniEGH17} consider the verification problem for semantics, which is equivalent to \wrasem, but with a model that allows unbounded numbers of pending messages.
They show the problem is \EXPSPACE-complete.

Another primary direction is verifying \emph{single runs}, expressed as sequences of memory access events, to determine whether that run satisfies a consistency model.
This problem is quite complex in general, as, for at least all the memory models studied in this work, there is an unbounded number of possible reorderings of the observed events from one local view to the other.
Several works consider the testing problem under the \rasem semantics \cite{DBLP:journals/pacmpl/AbdullaAJN18,DBLP:journals/pacmpl/TuncA0K0P23}.
%
%
These works show that the testing problem has polynomial complexity for \rasem.
Bouajjani et al. \cite{DBLP:conf/popl/BouajjaniEGH17} show that the testing problem for the \wrasem semantics has polynomial complexity.
Both works focus on providing a specific implementation or verifying a single run from an unknown implementation.

{\bf Contribution:}
We consider the complexity of the consistency verification problem under the \rasem semantics.
To state our results, we use the classical {\it register machine} model to describe the underlying implementation that handles memory access.
The model consists of a finite-state machine extended with a finite set of registers that store data values from an unbounded domain.
The machine interacts with a finite set of external threads through write (where the register machine inputs a value to a register) and read (outputting a stored value) operations performed on a finite set of variables.
%
%
Furthermore, the machine can perform internal transitions to transfer (i.e., copy) data  between registers.
As is standard in the literature for the type of architecture we are modeling, we do not allow data-dependent transitions.
The model is conceptually simple, providing a concise framework to state our complexity results.
At the same time, it is sufficiently expressive to model relevant
features needed to model cache protocols or distributed systems,
for example rendezvous communication, broadcasting fences, vector clocks,
broadcast communication and store buffers \cite{DBLP:conf/popl/BouajjaniEGH17,DBLP:journals/fmsd/Delzanno03}.
Moreover, recent works use automata-like formalisms for learning models of implementations and detecting bugs \cite{Fiterau_Uppsala_NDSS,Fiterau_Takvist_Automata_learning_TACAS}. Such works enhance the relevance of register machines for verifying program behaviors, such as consistency with weak memory.

%
Given a register machine, we consider the verification problem: decide whether,
 for \emph{all} interacting program, the register machine never produces
  a {\it bad behavior}, i.e., store and return values in
  a way that violates any one of the $\rasem$-family
  of models (namely, $\rasem$, $\wrasem$, and
  $\srasem$).
To do this, one must explore all possible runs of a given register machine.
The state space of the register machine is infinite (since the data domain is infinite), and the set of paths is also infinite, so the problem's decidability is not obvious.
Here, we show that it is decidable for all considered models, and in the case of non-polynomial complexity
also provide lower bounds. Our main contributions are the proofs that:
\begin{itemize}
    \item The verification problem for $\wrasem$ is in $\bigO{n^5}$ time
     (Section \ref{sec:algorithms}).

    \item The verification problem for the
     $\rasem$ and $\srasem$ semantics is \PSPACE~complete (Section \ref{sec:hardness}).
     %
\end{itemize}


%
%
%
%
The main body of our technical contribution lies in determining
a way to explore only finite (and finitely many) runs of
the register machine.
For our hardness results we provide a reduction from the
problem of emptiness of DFA intersection.
This work is an extension of the results reported
in \cite{RA_verification_2025}. The key
extra contributions lie in:
\begin{itemize}
  \item providing the full proofs for all our theorems and algorithms,
 \item closing the complexity gap for the verification problem
 $\rasem$ and $\srasem$ which was only known to be
  between \PSPACE~and both \NP and \coNP,
 \item extending our \PSPACE-hardness result to the
 memory models of \textit{parallel snapshot isolation} ($\psisem$) and \textit{sequential consistency} ($\scsem$), and
 \item providing more examples, figures and discussions,
 to improve on redability and presentation.
\end{itemize}

\section{Preliminaries}\label{sec:overview}

%
%
%
%
%
%
%
To formulate the verification problems we study, we need two formalisms that describe the platform and the consistency model.
We use register machines and  execution graphs respectively.
In what follows, we will use the following notation:
\begin{itemize}
    \item  Given a relation $\rel$, $\fundomof\rel$ denotes its domain; $\rclosureof\rel$ and $\tclosureof\rel$  denote its
reflexive and transitive closures; and $\invmkrelof\rel$
denotes its inverse.
    \item Given a function $\fun$, we write
    $\updatefun\fun{\xvar}{\yvar}$, to denote a new function $\fun'$, where $\fun'(\xvar)=\yvar$, and $
    \fun'(\xvar')=\fun(\xvar')$, if $\xvar' \neq \xvar$.
    \item Given an expression $S$, we denote as
$S(\sfrac{a}{b})$, the expression $S$, where all occurrences of $b$ have been replaced with $a$. Note that if $S$ had no occurrences of $b$ then $S(\sfrac{a}{b}) = S$.
    \item For a set $\mathcal{S}$ and an element $a$, we denote with $\mathcal{S}\oplus a$ the union of $\mathcal{S}\cup \{a\}$.
\end{itemize}

\subsection{Register Machines}
\label{rm:subsection}
\begin{figure}
\begin{subfigure}[h]{0.4\linewidth}
    \scalebox{0.8}{
      \begin{tikzpicture}[initial text=,]
        \node[state, initial] (q0) at (0, 0) {$q_0$};
        \node[state] (q1) at (3, 0) {$q_1$};

        \node[fill=lightgray] (lb) at (-2, 0) {$\mathcal{M}_1$};

        \draw[-stealth] (q0) to[bend left] node[above,midway] {$\wopof{\theta}{x}{a}$} (q1);
        \draw[-stealth] (q1) to[bend left] node[below,midway] {$\wopof{\phi}{x}{b}$} (q0);

        \draw[-stealth] (q0) to[loop above] node[above,midway] {$\ropof{\theta}{x}{b}$} (q0);
        \draw[-stealth] (q0) to[loop below] node[below,midway] {$\ropof{\phi}{x}{a}$} (q0);

        \draw[-stealth] (q1) to[loop above] node[above,midway] {$\wopof{\theta}{x}{a}$} (q1);
      \end{tikzpicture}
      }
  \end{subfigure}
  \begin{subfigure}[h]{0.4\linewidth}
    \scalebox{0.8}{
    \begin{tikzpicture}
      \node (s1) at (0,2) {$\tuple{q_0, 0, 0}$};
      \node (s2) at (4,2) {$\tuple{q_1, 1, 0}$};
      \node (s3) at (8,2) {$\tuple{q_1, 2, 0}$};
      \node (s4) at (8,0) {$\tuple{q_0, 2, 3}$};
      \node (s5) at (4,0) {$\tuple{q_0, 2, 3}$};
      \node (s6) at (0,0) {$\tuple{q_0, 2, 3}$};

      \node[fill=lightgray] (lb) at (4, 1) {$\rho$};

      \draw[-stealth] (s1) to node[above, midway] {$\wopof{\theta}{x}{a}\textcolor{red}{\sslash1}$} (s2);
      \draw[-stealth] (s2) to node[above, midway] {$\wopof{\theta}{x}{a}\textcolor{red}{\sslash2}$} (s3);
      \draw[-stealth] (s3) to node[left, midway]  {$\wopof{\phi}{x}{b}\textcolor{red}{\sslash3}$} (s4);
      \draw[-stealth] (s4) to node[below, midway] {$\ropof{\phi}{x}{a}\textcolor{red}{\sslash2}$} (s5);
      \draw[-stealth] (s5) to node[below, midway] {$\ropof{\theta}{x}{b}\textcolor{red}{\sslash3}$} (s6);
    \end{tikzpicture}
    }
  \end{subfigure}
  \caption{A register machine $\irmachine1$ and a run $\run$ of $\irmachine1$}
  \label{fig:rmachine_and_run}
\end{figure}

A {\it register machine}, or shortly a {\it machine}, is a finite-state automaton extended with a finite set of registers that store data values from an unbounded domain.
The machine performs input (write) operations and output (read) operations on a finite set of variables.
Read and write operations induce external actions that
synchronize the machine with its environment, i.e., with an external program consisting of a finite set of threads that run on the machine.
Figure \ref{fig:rmachine_and_run} depicts a register machine $\irmachine1$ with two states $\istate0$, $\istate1$, and two registers $\areg$ and $\breg$.
The machine $\irmachine1$ manages two threads $\thread$ and $\pthread$ accessing a (single) shared variable $\xvar$.
It starts executing from the initial state $\istate0$ with the initial register values $0$.
Each transition of the machine is labeled by an {\it operation}.
For instance, the transition label from $\istate0$ to $\istate1$ is the write operation $\wopof\thread\xvar\areg$.
Here, the machine $\irmachine1$ accepts a request from the thread $\thread$ to write a new value on the variable $\xvar$, upon which the machine stores the written value in the register $\areg$.
The machine allows the program running on it to choose the written value.
In $\istate1$, the machine loops performing a sequence of write operations as the one above.
The label of the transition from $\istate1$ to $\istate0$ is the write operation $\wopof\pthread\xvar\breg$, in which $\pthread$ performs a write operation, and $\irmachine1$ stores the written value in $\breg$.
In $\istate0$, the machine accepts read requests from the threads.
The operation $\ropof\thread\xvar\breg$ means that $\irmachine1$ accepts a request from the thread $\thread$ to read the value of the variable $\xvar$, upon which the machine returns the value currently stored in the register $\breg$.
We can explain the operation $\ropof\pthread\xvar\areg$ similarly.
In the general case, a register machine is meant to allow \textit{any} kind of request (i.e., a read or write from any thread to any variable) from the environment (program) at any time, no matter what state it is in.
Such a register machine will be called \textit{reactive}.
For example we give in Figure~\ref{WRA:aut:fig} the reactive version
of the register machine $M_1$ from Figure~\ref{fig:rmachine_and_run}.

\begin{figure}[]
    \centering{\begin{tikzpicture}[text centered]
      \node[state] (q0) at (0,6) {$q_0$};
      \node[state] (q1) at (-3,3) {$q_1$};
      \node[state] (q2) at (3,3) {$q_2$};

      \draw[-stealth] (q0) to[loop above] node[midway] {\begin{tabular}{c}$\ropof{\theta}{x}{b}$ \\ $\ropof{\phi}{x}{a}$\end{tabular}} (q0);
      \draw[-stealth] (q0) to node[midway, fill=white] {$\wopof{\theta}{x}{a}$} (q1);
      \draw[-stealth] (q0) to node[midway, fill=white] {$\wopof{\phi}{x}{b}$} (q2);
      \draw[-stealth] (q1) to[loop left] node[midway] {\begin{tabular}{r}$\wopof{\theta}{x}{a}$ \\ $\ropof{\theta}{x}{a}$\end{tabular}} (q1);
      \draw[-stealth] (q1) to[bend left=50] node[midway, fill=white]{$\wopof{\phi}{x}{b}$} (q0);
      \draw[-stealth] (q1) to[bend left=10] node[midway, fill=white] {$\ropof{\phi}{x}{b}$} (q2);

      \draw[-stealth] (q2) to[loop right] node[midway] {\begin{tabular}{l}$\ropof{\phi}{x}{b}$\\ $\wopof{\phi}{x}{b}$ \\ $\ropof{\theta}{x}{a}$\end{tabular}} (q2);
      \draw[-stealth] (q2) to[bend left=10] node[midway, fill=white] {$\wopof{\theta}{x}{a}$} (q1);


    \end{tikzpicture}
    }
\caption{A reactive version of the $\irmachine1$ of Figure
\ref{fig:rmachine_and_run}.
The labels $\ell_{i,j}$ describe transitions from
state $\qstate_i$ to $\qstate_j$. As we will see later,
both this version and the original $M_1$, satisfy $\wrasem$.}
\label{WRA:aut:fig}
\end{figure}
Formaly we have:
\begin{definition}
\label{rmachines:mode:subsection}
Assume a set $\threadset$ of threads, a set $\varset$ of variables, and a set $\regset$ of registers.
We assume that the variables and the registers range over a (potentially infinite) set $\valset$ of data values with the particular value $\zero\in\valset$.
A {\it \textbf{register machine}}
 $\rmachine$ is a tuple $\rmachinetuple$ where $\stateset$
  is the finite set of states, $\initstate\in\stateset$ is
  the initial state, and $\transitionset$ is the finite set of transitions.
A transition is a triple of the form $\tuple{\mystate,\op,\pqstate}$ where $\mystate,\pqstate\in\stateset$ are states, and $\op$ is an operation.
The operation $\op$ can be in one of the following three forms:
\begin{itemize}
\item
$\wopof\thread\xvar\areg$ receives the value of the variable $\xvar$ from $\thread$ and writes (stores) the value in register $\areg$.
The environment selects the written value (the program running on $\rmachine$).
\item
$\ropof\thread\xvar\areg$ reads of value of the variable $\xvar$ from the register $\areg$ and delivers the stored value to $\thread$.
\item
$\ceventof{\areg}{\areg'}$ copies the value stored in the register $\areg'$ to the register $\areg$.
\end{itemize}
\end{definition}

For any register machine, a depth-first search allows to detect any
unreachable state (from $\stateset$).
In the following, we thus assume that all states in $\stateset$ are reachable by at least some path.

 \begin{remark}
  Later on we will use the notation
 $\iqstate1\mkrelof{\text{Predicate}}\iqstate2$ , to denote that there
exists a transition in $\transitionset$ from $\iqstate1$ to $\iqstate2$,
which satisfies the predicate. For example,
  $\iqstate1\mkrelof{\thread \wedge \xvar}\iqstate2$
  is satisfied if there exists a transition
from $\iqstate1$ to $\iqstate2$, where the transition label uses
thread $\thread$ and variable $\xvar$.
 \end{remark}

\begin{remark}
  Even thought there are no explicit $\epsilon$ transitions in these machines,
  it is easy to simulate such transitions,
  by allowing for example a transition labeled by $\ceventof{\areg}{\areg}$.
  Later on we will be constructing machines which contain such $\epsilon$ steps.
\end{remark}


\subsection{Operational Semantics}
\label{op:sem:subsection}
%
%

We define the operational semantics of a register machine by defining
 the transition system it induces, i.e., by defining the set of configurations   of the machine together with a transition relation on them.\footnote{We use the term {\it transition} to refer both to the set of transitions
   in the syntax of the machine (Def. \ref{rmachines:mode:subsection})
    and to the transition relation on configurations. The meaning will always be clear from the context.}
A {\it configuration} $\conf$ is of the form $\tuple{\mystate,\regvalmapping}$ where $\mystate\in\stateset$ defines the state of the machine, and $\funtype\regvalmapping\regset\valset$ defines the value $\regvalmappingof\areg$ of each register $\areg\in\regset$.
The \textit{initial configuration} $\initconf$ is the pair $\tuple{\initstate,\lambda\,\regset.\,\zero}$, i.e., the machine $\rmachine$ starts running from a configuration where it is in its initial state and all its registers contain the value $\zero$.

For example, a configuration of the machine $\irmachine1$ from Fig.~\ref{fig:rmachine_and_run} is a triple $\tuple{\istate,\ii_a,\ii_b}$ describing the local state, and the contents of the registers $\areg$ and $\breg$.
In this example we see that a run $\run$ consists of a sequence of transitions.
The run starts from the initial configuration where $\irmachine1$ is in its initial state $\istate0$, and the registers contain their initial values $0$.
When executing a transition, we use the operation of the transition to generate an {\it action} describing an observable interaction between $\irmachine1$ and its environment.
\begin{figure}
\centering
\begin{gather*}
\scalebox{0.6}{
\threeinfer
    { \text{Write}}{
      \tuple{\mystate,\wopof\thread\xvar\areg,\pstate}
      \in\transitionset
      }
    {
      \val\in\valset\;\;\;\pregvalmapping=\updatefun\regvalmapping\avar\val
    }
    {
      \tuple{\mystate,\regvalmapping}
      \rmovesto\rmachine{\weventof\thread\xvar\areg\textcolor{red}{\sslash\val}}{}
      \tuple{\pstate,\pregvalmapping}
    }
}
\irspace
\scalebox{0.6}{\threeinfer{\text{Read}}
    {
       \tuple{\mystate,\reventof\thread\xvar\areg,\pstate}
      \in\transitionset
      }
    {
      \regvalmappingof\avar=\val
    }
    {
      \tuple{\mystate,\regvalmapping}
      \rmovesto\rmachine{\reventof\thread\xvar\areg\textcolor{red}{\sslash\val}}{}
      \tuple{\pstate,\regvalmapping}
    }
}\irspace
\scalebox{0.6}{\threeinfer{\text{Copy}}
    {
      \tuple{\mystate,\ceventof{\areg}{\areg'},\pstate}
      \in\transitionset
      }
    {
      \pregvalmapping=
      \updatefun\regvalmapping
          {\areg}{\regvalmappingof{\areg'}}
    }
    {
      \tuple{\mystate,\regvalmapping}
      \rmovesto\rmachine{\silentevent}{}
      \tuple{\pstate,\pregvalmapping}
    }
}
\end{gather*}
  \caption{The semantics of a register machine's three operations. Write and copy operations update the state of the memory $\regvalmapping$, while read operations only update the state of the register machine. }
  \label{fig:semantics_exec_graphs}
\end{figure}
We use $\confset$ to denote the set of configurations.
A transition is of the form $\iconf1\movesto\action\iconf2$ where $\iconf1,\iconf2\in\confset$ are configurations, and $\action$ is an operation augmented with a concrete value to be read or written.
We define the transition relation between configurations according to the inference rules of Fig.
 \ref{fig:semantics_exec_graphs}.
In the write rule, the machine executes a transition from $\mystate$ to $\pstate$
while processing a write operation.
The configuration changes state accordingly and updates the value of the relevant register as implied by the operation.
In the read rule, the machine processes a read operation that returns the relevant register's value.
%
%
The machine performs a register assignment operation in the copy rule.
The operation is not visible to the external threads; hence, it is labeled by the silent event $\silentevent$.
%

A run $\run$ of the program is
a sequence
 $\iconf0\movesto{\action_1}\iconf1\movesto{\action_2}\cdots\movesto{\action_n}\iconf\nn$
  of transitions,
where each $\action_i$ is one of the operations described
 in Figure \ref{fig:semantics_exec_graphs}.
%
%
%
 We say that $\run$ is \textit{\textbf{differentiated}} if, 
 %
%
 for any given variable $\xvar\in\varset$, the write events in $\run$ all use different values.

\subsection{Execution Graphs}
\label{egraph:overview:subsection}
We will be using execution graphs to represent a run, as well as to describe our
models in the classic axiomatic style \cite{shasha_snir_traces}.
The nodes of an {\it execution graph} ({\it egraph} for short)  are {\it events}.
Figure \ref{fig:ex_exec_graph} shows an example of an egraph (specifically
for the egraph that corresponds
to the run we saw in Figure \ref{fig:rmachine_and_run}).

\begin{figure}
  \centering
        \begin{tikzpicture}
          \node (w1) at (0,3.5) {$\eventof\wtype{\theta}{x}{1}$};
          \node (w2) at (0,1.8) {$\eventof\wtype{\theta}{x}{2}$};
          \node (w3) at (4,1.8) {$\eventof\wtype{\phi}{x}{3}$};
          \node (r1) at (4,0) {$\eventof\rtype{\phi}{x}{2}$};
          \node (r2) at (0,0) {$\eventof\rtype{\theta}{x}{3}$};

          \draw[-stealth, \pocolor] (w1) -- (w2) node[midway, fill=white] {$\porel$};
          \draw[-stealth, \pocolor] (w2) -- (r2) node[midway, fill=white] {$\porel$};
          \draw[-stealth, \pocolor] (w3) -- (r1) node[midway, fill=white] {$\porel$};
          \draw[-stealth, dashed]  (w2) -- (r1) node[midway, fill=white] {$\rfrel$};
          \draw[-stealth, dashed]  (w3) -- (r2) node[midway, fill=white] {$\rfrel$};
          \draw[-stealth ,dotted] (w2) to [bend left] (w3) node[above,yshift=0.2cm,fill=white] {$\corel_x$};
          \draw[-stealth ,dotted] (w3) to [bend right] (w2) node[above,xshift=0.1cm,yshift=0.25cm,fill=white] {$\corel_x$};

      \end{tikzpicture}
\caption{The egraph $\egraph$ of the run $\run$ produced by the register
machine $M_1$ of Figure \ref{fig:rmachine_and_run}. }
        \label{fig:ex_exec_graph}
\end{figure}

An event corresponds to an action performed by a register machine when interacting
 with its environment.
%
%
The egraph edges specify different relations on the events.
In this paper, to define our consistency models,
we will work with three binary relations (\cite{DBLP:journals/toplas/LahavB22}): (a) the
 {\it program-order relation} ($\porel$), depicted by solid edges,
  totally orders the events in each thread;
  (b) the {\it reads-from} relation ($\rfrel$), depicted by dashed edges,
  associates every read event with the write event it reads from;
  and (c) the {\it coherence-order} relation ($\corel$), depicted by dotted edges,
  partially orders the writes on each variable.
  Different consistency models are defined by forbidding different types of cycles in the egraph (as described in Section \ref{consistency:models:subsection} below).
We associate the runs of a register machine with egraphs.

\subsubsection{Definitions}
%
 An {\it event} $\event$ is of the form $\eventof\ttype\thread\xvar\val$ where $\ttype\in\set{\wtype,\rtype}$ is the type of the event (write or read), $\thread\in\threadset$ is the thread performing the event, $\xvar\in\varset$ is the variable on which $\thread$ conducts the event, and $\val$ is the value that is either written or read from memory.
We define $\typeattrof\event:=\ttype$, $\threadattrof\event:=\thread$, $\valattrof\event:=\val$,
 and $\varattrof\event:=\xvar$.
We will use a set $\initeventset=\setcomp{\xinitevent\xvar}{\xvar\in\varset}$ of {\it initial} write events, where $\xinitevent\xvar$ represents a dummy event writing the initial value $\zero$ to $\xvar$.
We assume that the initial events do not belong to any thread.
We use $\eventset$ to denote the set of all events.
For any set of events $\events\subseteq\eventset$, we define the relation
$\mkrelof\events:=\setcomp{\tuple{\event,\event}}{\event\in\events}$, i.e., it is the restriction of the identity relation to the set of events in $\events$.
For $\ttype\in\set{\wtype,\rtype}$, we define the relation $\mkrelof\ttype:=\setcomp{\tuple{\event,\event}}{\typeattrof\event=\ttype}$, i.e. it is the restriction of the identity relation to the set of events of type $\ttype$.
Similarly, for a thread $\thread\in\threadset$, we define the relation $\mkrelof\thread:=\setcomp{\tuple{\event,\event}}{\threadattrof\event=\thread}$.
Sometimes, we view these relations as sets and write, e.g.,  $\event\in\mkrelof\rtype$ to denote that $\event$ is of type $\rtype$.
We also consider Boolean combinations of these relations, so we write $\mkrelof{\events\land\rtype}$ to denote the set of events in $\events$ of type $\rtype$.
Fix a set $\events: \initeventset\subseteq\events\subseteq\eventset$ of events.
\begin{itemize}
    \item A {\it program-order} on $\events$ is a relation $\porel$ defined as a union $\cup_{\thread\in\threadset}\,\xporel\thread$ such
that  $\xporel\thread$ is a total order on the set of events in $\mkrelof{\events\land\thread}$.
    In other words, $\porel$ totally orders all the events in $\events$ belonging to each thread.
    \item  A {\it reads-from} relation $\rfrel\subseteq\mkrelof{\events\land\wtype}\times\mkrelof{\events\land\rtype}$ assigns to each read event $\revent$ a single write event $\wevent$ in $\events$ with $\varattrof\revent=\varattrof\wevent$ and  $\valattrof\revent=\valattrof\wevent$.
    We will write, $\wevent\mkrelof\rfrel\revent$ to mean that $\revent$ takes its value from $\wevent$.
    \item
    A {\it partial-coherence-order} on $\events$ is a relation $\partcorel$ defined
     as a $\cup_{\xvar\in\varset}\,\partcorel_\xvar$
     such that $\partcorel_\xvar$
      is a partial order on the set of write events on $\xvar$.
    We require that $\xinitevent\xvar$ is the smallest element in the sub-relation  $\partcorel_\xvar$.
    A {\it total-coherence-order}, $\corel$ on $\eventset$ is a coherence-order
    in which the $\xvar$-sub-relations are total.
    In other words, $\corel=\cup_{\xvar\in\varset}\,\xcorel\xvar$ and
     $\xcorel\xvar$ is a total order on the set of write events on $\xvar$.
    In this paper, we only use coherence-order relations that can be derived from the $\porel$- and $\rfrel$-relations.

\end{itemize}
%
%
We also define the {\it happens-before} relation $\hbrel:=(\porel\cup\rfrel)^+$.
%
A {\it partial execution graph} $\egraph$ is a tuple $\tuple{\events,\porel,\rfrel,\partcorel}$ where:
(i)
$\events\subseteq\eventset$ is a set of events,
(ii)
$\porel$ a program-order on the set $\events$,
(iii)
$\rfrel$ is a  reads-from relation on $\events$, and
(iv)
$\partcorel$ is a partial coherence-order relation on $\events$.
A {\it total execution graph} 
is a partial execution graph
 in which  the coherence-order relation is total.

%
The initial egraph is defined
 as $\initegraph:=\tuple{\initeventset,\emptyset,\emptyset,\emptyset}$,
  i.e., it only contains the initial events,
   and all its relations are empty. In what follows
   we will generally be checking that the register machines
   respect the memory models without reading the values of the initial
   events. However this is only a soft requirement and can be
   easily altered in the implementation.

\subsubsection{Adding Events}
We define an operation $\addevent$ that adds a new event
 to an egraph, according to the rules given in Figure \ref{add:event:fig}.
\begin{figure}
\centering
\scalebox{0.6}{
    \threelabelinfer
    {
    \text{Write Events}
    }
    {
      \event = \eventof\wtype\thread\xvar\val~~~ \events' = \events \cup \set{\event}
      }
    {
      \porel' = \porel \cup \set{(\pevent,\event) \mid \pevent \in \events \wedge \threadattrof\pevent=\thread}
    }
    {
      \tuple{\events,\porel,\rfrel,\partcorel} \movesto{\event}\tuple{\pevents,\porel',\rfrel,\partcorel}
    }{}
}
\scalebox{0.6}{
\fivelabelinfer
    {
    \text{Read Events}
    }
    {
      \event = \eventof\rtype\thread\xvar\val\irspace\exists \pevent\in\events:~ \typeattrof\pevent:=\wtype ~~~ \varattrof\pevent=\xvar ~~~ \valattrof\pevent=\val
      }
    {
     \pevents =\events \cup \set{\event} ,~~~ \rfrel' = \rfrel \cup \set{(\pevent,\event)}
    }
    {
      \porel' = \porel \cup \set{(\ppevent,\event) \mid \ppevent \in \events \wedge\threadattrof{\ppevent}=\thread}
    }
    {
    \partcorel'=\partcorel \cup \set{(\ppevent,\pevent) \mid \typeattrof{\ppevent}=\wtype \wedge \varattrof{\ppevent}=\xvar\wedge {\ppevent} \in \events \wedge \ppevent \tclosureof{\mkrelof{\porel \cup \rfrel}} \event}
    }
    {
    \tuple{\events,\porel,\rfrel,\partcorel} \movesto{\event}\tuple{\pevents,\porel',\rfrel',\partcorel'}
    }{}
  }
\caption{The rules for adding events to an egraph.
 A write event only causes an update to the program order relation
  $\porel$, while a read event also creates $\rfrel$ and $\corel$ edges.
  For the $\corel$ update we consider only $\ppevent \neq \pevent$.
  Since $\corel$ edges are added only when necessary,
  the resulting $\corel$ is a partial one.}
\label{add:event:fig}
\end{figure}
If the new event $\wevent$ is a write event performed by a thread $\thread$, we add $\wevent$ to the set of events.
%
%
%
Adding a write event does not affect the $\rfrel$ and $\corel$ relations.
The $\porel$ is updated to an edge from all existing events
of the thread $\thread$ to the new $\wevent$.
In subsequent figures we will be drawing $\porel$ edges only between
``consecutive'' events (of any type).

If the new event to be added is a read event $\revent$, then we also need to provide a write event $\wevent$ that already belongs to $\events$ from which $\revent$ will read its value.
The events $\revent$ and $\wevent$ should have identical variables and values.
We modify the $\porel$-relation in the same manner as we did for write events.
We modify the $\rfrel$-relation by adding the new pair $\tuple{\wevent,\revent}$ indicating that $\revent$ is reading from $\wevent$.
Finally, we update the partial coherence order so that we maintain
 the invariant
that the latter is a \textit{modification order} as was
 defined in \cite{DBLP:conf/popl/LahavGV16}.
A modification order relation esentially orders
write accesses per variable, so that reads are reading values
only
by writes that happened before them (via the $\hbrel$ relation) and
are maximal according to the modification order.
To that end, we consider $\wevent$ and $\revent$ to be
sources and targets, and then search for write events
$\pwevent$ that are $\hbrel$-before $\revent$ to connect to $\wevent$.
%

%

\subsubsection{From Runs to Egraphs}
We associate to each run $\run$ of a register machine a corresponding
egraph $\egraph:=\mkegraphof\run$.
To do so we will need to match the observable register machine operations
 (not copies) $\op$ with egraph events $\event$.
For an operation $\op=\op\sslash\val$
(from Fig. \ref{fig:semantics_exec_graphs}),
 the corresponding egraph event inherits the type, thread, variable,
  and value of $\op$, but not the register. %
  Thus,

  \begin{definition}\label{def:machine_event_to_graph_event}
    For an  operation $\op\sslash\val$ of a register machine,
     we define:
    $\event(\eventof\ttype\thread\xvar\areg\sslash\val )= \eventof\ttype\thread\xvar\val$.

  \end{definition}We define the operator $\mkegraph$ inductively as follows:

\begin{definition} For a run $\run$, we define
    \begin{itemize}
        \item $\mkegraphof{\run\movesto{\op\sslash\val}\iconf\nn}%
        = \mkegraphof{\run}\addevent \event(\op\sslash\val)$.
        \item $\mkegraphof{\run\movesto{\tau}\iconf\nn}%
        = \mkegraphof{\run}$
        \item $\mkegraphof\epsilon = \initegraph$.
    \end{itemize}
\end{definition}

Note here that the egraph corresponding to a run will not be a \textbf{total} one in the general case, as for example the run might not include any $\revent$, which means we will have no $\corel$ edges.
Moreover, this construction is \textbf{non-deterministic}, as for example, there might be several write events with the same value that could be read in a read event $\revent$. However as we have defined our register machines to be \textit{\textbf{data independent}} we will soon see this is not a problem for the verification process, and in fact we will  have a single execution graph per run.
Finally, it is possible that $\mkegraphof{\run}$ fails in some step as for example there might be no source write event for a given read. Our definitions over consistency models assume this does not happen, and later on the algorithm will indeed check this separately.
\subsection{Consistency Models}
\label{consistency:models:subsection}
\begin{figure}
  \centering
  \begin{tikzpicture}
      \node (w1) at (7,2) {$\eventof\rtype{\theta_1}{x}1$};
      \node (r2) at (7,0) {$\eventof\wtype{\theta_1}{x}2$};
      \node (w2) at (11,2) {$\eventof\rtype{\theta_2}{x}2$};
      \node (r1) at (11,0) {$\eventof\wtype{\theta_2}{x}1$};

      \draw[-stealth, \pocolor] (w1) -- (r2) node[midway, fill=white] {$\porel$};
      \draw[-stealth, \pocolor] (w2) -- (r1) node[midway, fill=white] {$\porel$};
      \draw[-stealth, dashed]  (r1) -- (w1) node[near start, fill=white] {$\rfrel$};
      \draw[-stealth, dashed]  (r2) -- (w2)  node[near start, fill=white] {$\rfrel$};
  \end{tikzpicture}
\caption{An execution graph that contains a cycle on $\hbrel$ }
\label{fig:cyclic_hb}
\end{figure}
A declarative memory model is formulated as a collection of constraints on execution graphs, which
determine the consistent execution graphs —~the ones allowed by the model.
In this section, we will formulate the three consistency models ($\anysem$) we work with.
All mentioned memory models are weaker than Sequential Consistency (SC), and allow for less restrictive memory accesses. SC requires that as soon as some value has been written in some variable, this is immediately visible to all threads.
The models we study instead allow for several threads to still view older values written in the variable, until they become ``aware'' of a new write on some path that ``hides'' the old value.

We define a consistency model $\anysem$ by forbidding different forms of cycles in egraphs.
All the consistency models we consider in this paper require the $\hbrel$-relation to be acyclic, i.e., the transitive closure of $\porel\cup\rfrel$ is a (strict) partial order.
For instance, the egraph of Fig. \ref{fig:cyclic_hb} contains a cycle on $\hbrel$  and hence it does not satisfy any of our consistency models.
Besides this condition on $\hbrel$, our consistency models impose additional constraints on the egraph  \cite{DBLP:conf/popl/LahavGV16, DBLP:journals/toplas/LahavB22}.

\label{sec:weak_memory_models}
\begin{definition}\label{def:egraph_satisfies_memory_model}
    Let $\egraph=\tuple{\events,\porel,\rfrel,\partcorel}$ be a (possibly partial) execution graph.
    \begin{itemize}
        \item We write $\egraph\models\wrasem$ to denote that for any
        write event $\wevent$ in $\egraph$,
         the relation $\mkrelof{\wtype\land\varattrof\wevent}\cdot\hbrel\cdot\wevent\cdot\hbrel\cdot\invmkrelof\rfrel$ is acyclic.
        \item We write $\egraph\models\rasem$ to denote that
         $\tclosureof{\mkrelof{\porel\cup\rfrel\cup\partcorel_{\xvar}}}$
         is irreflexive for each variable $\xvar\in\varset$.
        \item We write $\egraph\models\srasem$ to denote that
         the relation $\tclosureof{\mkrelof{\porel\cup\rfrel\cup\partcorel}}$
          is acyclic.
    \end{itemize}
\end{definition}
%
For a set  $\events$  of events, a program-order relation $\porel$ on $\events$, and a reads-from relation $\rfrel$ on $\events$, we write $\tuple{\events,\porel,\rfrel}\models\rasem$ if there is a total coherence-order relation $\corel$ on $\events$ such that $\tuple{\events,\porel,\rfrel,\corel}\models\rasem$ (rsp. $\srasem$, and $\wrasem$).
We write $\tuple{\events,\porel}\models\rasem$ if there is a reads-from relation $\rfrel$ and a coherence-order relation $\corel$ on $\events$ such that $\tuple{\events,\porel,\rfrel,\corel}\models\rasem$ (rsp. $\srasem$, and $\wrasem$).
\begin{definition}[Memory models over runs]
\label{def:run_satisfies_memory_model}
    For a run $\run$, we write $\run\models\rasem$ to denote that $\mkegraphof{\run} \models \rasem$
    (rsp. $\srasem$, and $\wrasem$).
\end{definition}
\textbf{Differentiated runs: } As we highlighted earlier, so far there might be several execution graphs associated with a single run. We write $\run\xmodels\difftype\rasem$
 if $\run$ is differentiated. Note that for differentiated runs, the transition rules of Figure \ref{add:event:fig} become deterministic.
%
%
We prove:
\begin{lemma}\label{lem:RA_to_diff_RA}
For each (general) run $\run$ of a register machine $\rmachine$, there exists a differentiated run $\run^\difftype$ such that
$\run\models\rasem$ iff $\run^\difftype\models\rasem$ (rsp. $\srasem$, and $\wrasem$).
\end{lemma}

\begin{proof}
    We observe that the register machine model is
    inherently data independent~\cite{DBLP:conf/popl/BouajjaniEGH17, POPL_Wolper86_data_independence}.
    We argue that if a non-differentiated run exists that violates the $\rasem$, then so does a differentiated one.
    To demonstrate this, for any arbitrary
     $\rmachine$ and $\run \in\runsetof\rmachine$, 
       we assume $\mkegraphof{\run} \not\models\rasem$.
    Note that currently the operator $\mkegraph$ applied to $\run$ might yield several output graphs (non-deterministically).
    We now create a new run $\run'$, by defining a meta-counter of steps taken in $\run$, and append that counter followed by a $\#$ to all values written in the memory.
    Since the domain of the register machine is infine we can encode this value
    within the existing alphabet of the register machine, and thus not alter the domain.

    The new run is now differentiated, as each value has a unique prefix,
     which means that $\mkegraphof{\run'}$ is now a function,
      and the register machine $\rmachine$ must take identical transitions
       while processing it (since it is data independent).
    Thus we would get that the resulting execution graph is identical to $\egraph_{\run}$, and thus also would be a violating one.

     In the converse direction, if a differentiated run violates a semantic model, the same run is a valid run of the register machine without the meta-variable modification.
    Thus $\rmachine\models\rasem$ iff $\rmachine\xmodels\difftype\rasem$ (rsp. $\wrasem$ or $\srasem$), and we are done.
\end{proof}

The above can be similarly extended for $\wraconsprob$ and $\sraconsprob$.

For the remainder of this paper \textbf{we only consider differentiated runs} against any of our memory models.
We also prove that for a (differentiated) run it suffices to check the partial execution graph that is formed by the rules of Figure \ref{fig:semantics_exec_graphs} against our memory models.

\begin{lemma}\label{lem:partial_co_to_total_co}
    Assume a register machine $\rmachine$ and $\run \in\runsetof\rmachine$. Then, for \\  $\mkegraphof{\run}  = \tuple{\events,\porel,\rfrel,\partcorel}$, if $\mkegraphof{\run} \models \rasem$, then there exists a total coherence order $\corel$, with $\partcorel \subseteq \corel$, such that $\tuple{\events,\porel,\rfrel,\corel} \models \rasem$. (rsp. $\srasem$, and $\wrasem$).
\end{lemma}

\begin{proof}
    Assume a register machine $\rmachine$ and a run $\run \in\runsetof\rmachine$,
    and let $  \mkegraphof{\run} =\egraph_{\run} = \tuple{\events,\porel,\rfrel,\partcorel}$ for some partial execution graph $\egraph_{\run}$.
    Let $\wevent, \wevent'$ be two write events on the same variable $\xvar$, such that neither $(\wevent,\wevent'), (\wevent',\wevent)$ are in $\partcorel_x$.
    If either $(\wevent,\wevent')$, or $ (\wevent',\wevent) \in\mkrelof{\partcorel_{\xvar} \cup \porel\cup \rfrel}^+$, then we append the pair to $\partcorel$ edge that respects the direction of this path and we are done.
    Otherwise, if $\wevent, \wevent'$ are incomparable w.r.t. $\mkrelof{\partcorel_{\xvar} \cup \porel\cup \rfrel}^+$ then we pick such events that are locally maximal for $\tclosureof{\mkrelof{\porel\cup\rfrel\cup\partcorel_{\xvar}}}$.
    In this case, it is safe to pick at random one of $(\wevent,\wevent'), (\wevent',\wevent)$ and add to $\partcorel_{\xvar}$ as no cycles will be introduced.
    We continue the above process until we are left with a total $\corel$ for each variable $\xvar$, and we are done.
    The construction for $\wrasem$ and $\srasem$ is identical,
    with the modification that we consider
     $\mkrelof{\wtype\land\varattrof\wevent}\cdot\hbrel\cdot\wevent\cdot\hbrel\cdot\invmkrelof\rfrel$, and $\mkrelof{\partcorel\cup \porel\cup \rfrel}^+$ paths, respectively.
\end{proof}
\begin{definition}[Memory models over register machines]
\label{def:rmachine_satisfies_memory_model}
For a register machine $\rmachine$, we write $\rmachine\models\rasem$ if $\forall\run\in\runsetof\rmachine.\;\run\models\rasem$.      (rsp. $\srasem$, and $\wrasem$). We will refer to the problem of determining whether a register machine satisfies these semantics as $\raconsprob$. (rsp. $\wraconsprob$ and $\sraconsprob$).
\end{definition}
In other words, a register machine satisfies the (\texttt{W}, \texttt{S})$\rasem$-semantics whenever all its runs do so.

\section{Algorithmic results for \wrasem}

We are now ready to state our results for the $\raconsprob$,
$\wraconsprob$, and $\sraconsprob$ problems. 
The results are stated in
Theorem \ref{thm:complexity_results}.
%
%
%
%
%
We prove that all the above problems are decidable.
 In the case of $\wrasem$ the complexity is polynomial
  to the size of the register machine, which is defined a sthe number of transitions it contains,
  while for $\rasem$ and $\srasem$ it is \PSPACE~complete.
For $\wraconsprob$ we prove it is sufficient to keep a constant amount of
information regarding paths in memory, and thus manage
 polynomial time complexity.
  In the case of $\raconsprob$ and $\sraconsprob$ the size of the data structures
   increases, which implies an increment of the number amount of possible paths.
   In these cases it is not sufficient to keep in memory only a constant number of information
   for each paths we are exploring, but instead a polynomial one, which still yields a \PSPACE~algorithm,
   but the time complexity is no longer polynomial.

%
%

\begin{theorem}\label{thm:complexity_results}
    For a given a register machine $\rmachine$ of size $n$:
    \begin{itemize}
        \item $\raconsprob$ is $\PSPACE$~complete.
        \item $\sraconsprob$ is $\PSPACE$~complete.
        \item $\wraconsprob$ is in $\bigO{n^{5}}$.
    \end{itemize}
\end{theorem}

The remainder of this section is dedicated to describing the idea of our algorithm(s), their correctness, and the hardness results.

\subsection{Algorithmic method for $\wraconsprob$}\label{sec:algorithms}

Our \P-TIME algorithm for $\wrasem$ 
 takes place in three modules.
All modules are based on a type of backwards reachability approach. %
The different modules start from potential violations of the condition they correspond to
and try to find paths leading to this violation by backtracking to the initial state.
We essentially keep track of what states have been
marked as origins of such potential violating paths, and on each iteration we process edges
leading to such states and propagate relevant information of existing paths backwards.

The reason for choosing to have a backward search is that it handles
 the copy operation better.
We can apply the standard weakest pre-condition operator (see rule $\largebluecircled{9}$ in Fig. \ref{fig:rules_wra}) to maintain optimal complexity.
Having a forward search would lead to exponential branching over equivalence classes (where registers with identical values are kept in the same equivalence class).
Our modules concern the following correctness aspects for all runs of a register machine $\rmachine$:

\begin{itemize}
\item $\rmachine$ does not allow for \textit{ghost reads}. Those are read events that can read the value of a register that is empty.
\item $\rmachine$ does not allow for \textit{mismatched variable reads}.
Those are read events on a variable $\xvar$, that read the value of a register that has been storing a value for a different variable $\yvar$.
\item $\rmachine$ does not allow for cyclic variable edit dependencies.
This last condition is the one that truly defines the $\wrasem$ semantics, as the other two are conditions that are required so that an execution graph can be formed.
\end{itemize}

The above conditions are all implied by the execution graph semantics stated in  Figure \ref{add:event:fig}. However, in order to simplify our algorithm we check them separately.
In this way when we get to the most difficult condition of the above, which is the third, we do not need add extra checks for the previous ones.
For example, consider having encountered some
register $\areg$ being used in some execution,
and we are storing some information about this,
 we do not need to also keep track of which variable $\xvar$
 was stored in $\areg$, since we have already confirmed that
 there are no executions allowing for operations on other variables
  to access registers that do not correspond to $\xvar$.
We formalize the above statements as:
\begin{proposition}
\label{prop:algo_RA_three_parts}
Given a register machine $\rmachine = \rmachinetuple$, we have that
$\rmachine\xmodels\difftype\wrasem$ iff:

\begin{enumerate}

\item For all $\run\in\runsetof\rmachine$, $\run = \run' \revent$, with $\revent = \outopof\thread\xvar\areg$,  there exists an event $\event \in \eventset$, such that $\run = \run_0\cdot \event\cdot \run_1$, and $\event$ is either a copy or write event that targets register $\areg$. \label{q:ghost_reads}

\item For all $\run\in\runsetof\rmachine$, $\run = \run'\revent$, with $\revent = \outopof\thread\xvar\areg$, there exists an event $\wevent_{\xvar} \in \eventset$, such that $\run = \run_0\cdot \wevent_{\xvar} \cdot \run_1$, and $\wevent_{\xvar}$ writes some value $\val$ on register $\breg$ (possibly $\breg =\areg$), and $\val$ is the last value assigned to $\areg$ during $\run_1$. \label{q:mismatched_variables}

\item For each run $\run\in\runsetof\rmachine$, with $\egraph_{\run} =
 \mkegraphof{\run} = \tuple{\events,\partcorel,\rfrel,\corel}$,  $\mkrelof{\wtype\land\varattrof\wevent}\cdot\hbrel\cdot\wevent\cdot\hbrel\cdot\invmkrelof\rfrel$ is acyclic for each write event $\wevent$ of $\events$. \label{q:RA_cycles}
\end{enumerate}
\end{proposition}

\begin{proof}
This proposition is counting on the fact that each run
is composed into an execution graph
through the $\mkegraphof{\run}$ operation,
and thus all necessary updates for each event
 in $\run$ have taken place in the graph
  by the time the last event is outputted.
We start then from a run $\run \in \runsetof\rmachine$,
with $\run = \run'\cdot \revent$,  $\revent = \outopof\thread\xvar\areg$.
Since $\rmachine\models \wrasem$ then $\run \models \wrasem$,
and thus $\mkegraphof{\run} = \egraph$, with $\egraph = \tuple{\events,\porel,\rfrel,\corel}$.

\textbf{Questions \ref{q:ghost_reads} and \ref{q:mismatched_variables}}. 
We have that $\egraph$ is the outcome of performing the
 update $\mkegraphof{\run'}\addevent \event(\revent\sslash\val)$.
  Clearly, as seen in Figure
   \ref{add:event:fig}, for this
   update to take place there exists event
    $\wevent =\inopof\thread\xvar\val \in \events$,
     with $\val$ being the unique value
      (all runs are differentiated due to Lemma \ref{lem:RA_to_diff_RA})
       that was outputted in $\revent$.
  Assume now that the value $\val$ was never written
   or copied to register $\areg$. Since all values are
   unique it is impossible for $\rmachine$ to perform the
   update $\rmovesto\rmachine\revent$.
  Therefore it must be the case that either
  the value of some $\breg$ has been copied to
   $\areg$ during $\run_1$, or that $\areg = \breg$, for some earlier event
   $\inopof\thread\xvar\breg\val$ and thus there
   exists event $\event \in \eventset$, such
   that $\run = \run_0\cdot \event\cdot \run_1$, and $\event$ is
    either a copy or write event that targets register $\areg$.
    This proves both questions as $\varattrof\wevent:=\xvar$.

\textbf{Question \ref{q:RA_cycles}} This follows from Lemma \ref{lem:RA_to_diff_RA}.

The reverse direction holds trivially as the $\wrasem$ problem is defined only as the acyclicity condition of  $\mkrelof{\wtype\land\varattrof\wevent}\cdot\hbrel\cdot\wevent\cdot\hbrel\cdot\invmkrelof\rfrel$ for each write event $\wevent$ of $\egraph_{\run}$.
\end{proof}

  The algorithms for checking the first and second condition for Proposition
  \ref{prop:algo_RA_three_parts}
  (with algorithms called ghost-read and mismatched-var modules respectively)
  are given first.
  Both problems can also be solved with a classical reachability analysis
  (with complexity not higher than checking the third condition).
We nonetheless add them in this paper, along with the proofs of their correctness,
 as an easy demonstration of our method, which can also hopefully
 prepare the reader for the final algorithm module and its proof. 

\subsection{Ghost reads}
Here we give the algorithm and the proof of correctness
for detecting that for every read there is a corresponding write,
 i.e.  Question \ref{q:ghost_reads} from \ref{prop:algo_RA_three_parts}.
We note that this particular question could easily be checked in polynomial time with a traditional breadth-first-search approach.
However, we chose to employ the backwards method that will be necessary later on from this stage already, so we can demonstrate the necessary concepts.
Our algorithm works as described in Figure \ref{algo:ghost_reads_pseudo}, and uses the following data structures:

\begin{figure}
\[
\begin{array}{c}
\text{Initialization and termination} \\
\begin{array}{c}
\scalebox{0.6}{
\onelabelinfer{\iqstate1\mkrelof{\rtype\land\thread\land\areg\land\xvar}\iqstate2
}{
  \regs(\qstate_1)\oplus\tuple{\areg}
}{$\largebluecircled{1}$}}
\end{array}
\begin{array}{c}
\scalebox{0.6}{
\onelabelinfer{
 \tuple{\areg}\in \regs(\qstate_{init})
}{
  \failconf\oplus\qstate_{init}
}{$\largebluecircled{2}$}}
\end{array}
\\
\text{Copied registers} \\
\begin{array}{c}
\scalebox{0.6}{
\onelabelinfer{\iqstate1\mkrelof{\areg\assigned\breg}\iqstate2 \irspace
 \tuple{\breg}\in \regs(\qstate_2)
}{
  \regs(\qstate_1)\oplus\tuple{\breg}
}{$\largebluecircled{3}$}}
\end{array}
\begin{array}{c}
\scalebox{0.6}{
\onelabelinfer{\iqstate1\mkrelof{\areg\assigned\breg}\iqstate2 \irspace
 \tuple{\areg}\in \regs(\qstate_2)
}{
  \regs(\qstate_1)\oplus\tuple{\breg}
}{$\largebluecircled{4}$}}
\end{array}
\\
\text{Transparency} \\
\begin{array}{c}
\scalebox{0.6}{
\onelabelinfer{\iqstate1\mkrelof{\rtype\land\thread}\iqstate2 \irspace
 \tuple{\breg}\in \regs(\qstate_2)
}{
  \regs(\qstate_1)\oplus\tuple{\breg}
}{$\largebluecircled{5}$}}
\end{array}
\begin{array}{c}
\scalebox{0.6}{
\onelabelinfer{\iqstate1\mkrelof{\wtype\land\thread\land\areg}\iqstate2 \irspace
 \tuple{\breg}\in \regs(\qstate_2)
}{
  \regs(\qstate_1)\oplus\tuple{\breg}
}{$\largebluecircled{6}$}}
\end{array}
\\
\begin{array}{c}
\scalebox{0.6}{
\onelabelinfer{\iqstate1\mkrelof{\areg\assigned\breg}\iqstate2 \irspace
 \tuple{\creg}\in \regs(\qstate_2)
}{
  \regs(\qstate_1)\oplus\tuple{\creg}
}{$\largebluecircled{7}$}}
\end{array}
\end{array}
\]
\caption{The inference rules for updating the $\regs$ data structure, for each state $\qstate \in \stateset$ of a register machine $\rmachine$. These rules constitute the implementation of the ghost-reads module. We assume that $\breg\neq\areg$.}
\label{fig:ghost_reads}
\end{figure}

\begin{figure}
  \begin{algorithm}[H]
  \KwIn{$\rmachine$}
   $rule\_updated := \true $\;
   tuples := $\emptyset$\;
   \For{$read\_edge \in \rmachine$}{
   tuples := tuples $ \cup~Rule1.(read\_edge)$\;
   }
   \While{$rule\_updated$}{
    $rule\_updated := \false$\;
    \For{ Rule $ \in$  Fig. \ref{fig:ghost_reads}}{
      \If{Rule.$condition() \in tuples$}{
        \eIf{Rule.update() ==~$\failconf$}{
        \KwRet{$\failconf$}\;
        }{
          $tuples = tuples~\cup$ Rule.$update()$\;
          $rule\_updated := \true$\;
        }
      }
    }
  }
   \caption{The algorithm for the Ghost-Reads Module}\label{algo:ghost_reads_pseudo}
  \end{algorithm}
  \end{figure}

\begin{itemize}
\item
It uses a data structure $\funtype{\regs}{\stateset}{2^{\regset}}$.
If $\tuple{\areg}\in \regs(\qstate)$ then there is a
 run $\run$ of $\rmachine$ starting from $\qstate$,
  in which whatever was stored in $\areg$ at $\qstate$ will be outputted
  along $\run$.

\item It first isolates all output edges $E_{out}$ of the register and then searches backwards through the edges or $\rmachine$. We update the data structure as Figure \ref{fig:ghost_reads} indicates.

\item If at any point a register $\areg$ becomes relevant to the initial state $\qstate_{init}$ we associate the $\failconf$ configuration to  $\qstate_{init}$. It follows that if while at $\qstate_{init}$ there exists a run where the register $\areg$ value is will be output as it was in $\qstate_{init}$, then it is possible to run the register machine and output from $\areg$ without initialising it. \end{itemize}

Intuitively, Figure \ref{fig:ghost_reads} associates a register $\areg$ to state $\qstate$ when there is a run from $\qstate$ outputting on some output edge the value of $\areg$ exactly as it was on $\qstate$.
These registers are propagated backwards to states $\pqstate$ that can reach $\qstate$ depending on the edge that connects $\pqstate$ to $\qstate$.
A register is not propagated backwards to states $\pqstate$ that can reach $\qstate$ if the transition leading to $\qstate$ is an input on $\areg$, as this means the value of $\areg$ is now overwritten.
Upon termination this $\bigO{n^3}$ algorithm detects whether a register machine $\rmachine$ can output a ghost value.
\begin{figure}
    \begin{tikzpicture}
        \node[state, initial] (q1) at (0, 0) {$q_0$};
        \node[state] (q2) at (3, 0) {$q_1$};
        \node[state] (q4) at (6, 0) {$q_2$};
        \node[state] (q5) at (9, 0) {$q_3$};

        \draw[-stealth] (q1) to[bend left] node[above,midway] {$\inopof{\theta_1}{x}{a_1}{-}$} (q2);
        \draw[-stealth] (q4) to[bend left] node[above,midway]
        {$\outopof{\theta_1}{x}{a_2}$} (q5);

        \draw[-stealth] (q2)to[bend left] node[above,midway]
        {$\outopof{\theta_2}{x}{a_1}$} (q4);
        \draw[-stealth] (q4)to[bend left] node[below,midway]
        {$\inopof{\theta_2}{y}{a_2}{-}$} (q1);
    \end{tikzpicture}
    \caption{%
    A register machine that produces the pattern of Figure \ref{fig:ex_exec_graph}.
     It satisfies $\wrasem$, but not $\rasem$,
     as there is a cycle of $\xcorel{x}$ edges.}
    \label{fig:ghost_read_violation}
\end{figure}
\begin{example}
Consider the register machine of Figure \ref{fig:ghost_read_violation}.
One can easily see that this register machine can produce a run
$\run = \qstate_0\movesto{\inopof{\theta_1}{x}{a_1}{-}}\qstate_1\movesto{\outopof{\theta_2}{x}{a_1}}\qstate_2\movesto{\outopof{\theta_1}{x}{a_2}}\qstate_3$, which outputs the value of $\areg_2$, even though nothing has been inputted on $\areg_2$ during $\run$.
Our algorithm from Figure \ref{fig:ghost_reads} would first perform $\regs(\qstate_2)\oplus\tuple{\areg_2}$, due to the
$\outopof{\theta_2}{x}{a_2}$ edge that would trigger rule $\bluecircled1$, and then trigger rules $\bluecircled5$ and $\bluecircled6$ in this order to perform the updates $\regs(\qstate_1)\oplus\tuple{\areg_2}$ and $\regs(\qstate_0)\oplus\tuple{\areg_2}$ respectively which would finally trigger rule $\bluecircled2$ and flag the $\failconf$ configuration.
\end{example}

We now give the proof of correctness for the ghost-read module.

\begin{proof}[Proof: Algorithm of  Figure \ref{fig:ghost_reads}]
We claim that for a register machine $\rmachine$,  $\failconf\oplus\initstate$ iff
there exists a run $\run\in\runsetof\rmachine$, $\run = \run'\cdot \revent$, with $\revent = \outopof\thread\xvar\areg$, and, there is no event $\wevent = \inopof\thread\xvar\areg\val \in \eventset$, such that $\run = \run_0\cdot \wevent\cdot \run_1$, where the value of $\wevent$ is outputted on $\outopof\thread\xvar\areg$

$\Rightarrow$ To prove the algorithm does not calculate false negatives, we claim the following:

If $\tuple{\areg}\in \regs(\qstate)$, then there exists a run $\run$, starting from $\qstate$ such that the value of $\areg$ is outputted upon some read event at the end of $\run$.
We will prove this by induction on the number of rules that have been applied (number of steps in the algorithm).

\begin{itemize}
\item \underline{Base case $\#steps=1$.} Since we start with no registers associated with any states, the only update that can take place is $\bluecircled{1}$.
For this rule to be applied, we must be at a state $\qstate$ that can output the value of a register within $1$ step, and those are the ones where there exists edge $\qstate\movesto{\outopof\thread\xvar\areg}\pqstate$ in $\transitionset$ of $\rmachine$. Thus, based on Figure \ref{fig:ghost_reads}, $\tuple{\areg}\in \regs(\qstate)$ and $\run = \outopof\thread\xvar\areg$.
\item \underline{Inductive hypothesis: the statement holds after $n$ steps of the algorithm.}
\item \underline{Inductive step:} On step $n+1$ we are possibly able to apply any of the rules in  Figure \ref{fig:ghost_reads}. $\bluecircled{1}$ has already been covered and thus we examine the following cases:

\textbf{Rule $\bluecircled{2}$ :} This rule does not update the registers associated to a state, thus from inductive hypothesis the claim holds.

\textbf{Rule $\bluecircled{3}$ :} We see that in order to apply this rule it must be the case that for the previous iteration of the algorithm $\tuple{\breg}\in \regs(\pqstate)$, and there exists edge $\qstate\movesto{\areg\assigned\breg}\pqstate$ in $\transitionset$.
 We see that if $b$ was able to be outputted after a run $\run$, then it will still be able to be outputted after the run $\event\cdot\run$, with $\event = \areg\assigned\breg$, and the statement holds.

\textbf{Rule $\bluecircled{4}$ :}
Very similarly as above, we see that the value of $\areg$ is not replaced with $\breg$ and thus the value of $\breg$ will be able to be outputted after $\event\cdot\run$, starting from $\qstate$.

\textbf{Rules $\bluecircled{5}$, $\bluecircled{6}$ and $\bluecircled{7}$ :} It is clear that the transitions studied here do not affect the fact that the value of a register $\breg$ can be outputted after the run guaranteed to exist from the inductive hypothesis, extended at the beginning by the relative event. Thus the claim holds.
\end{itemize}

We now clearly have that if $\failconf\oplus\initstate$ there exists a run $\run$ that, starting on $\initstate$, will output the value of $\areg$.
 However since $\initstate$ is where the register machine starts its computation with an empty set of registers, we have that $\rmachine$ would allow this run to take place even when $\areg$ is unitialized, and  the statement holds.

$\Leftarrow$ To prove the reverse direction, i.e. that the algorithm detects all violations, we claim the following:

If there exists a run $\run\in\runsetof\rmachine$, $\run = \run'\cdot \revent$, with $\revent = \outopof\thread\xvar\areg$, and, there is no event $\wevent = \inopof\thread\xvar\areg\val \in \eventset$, such that $\run = \run_0\cdot \wevent\cdot \run_1$, where the value of $\wevent$ is outputted on $\outopof\thread\xvar\areg$, then  $\failconf\oplus\initstate$.

To prove this we have: we know that since $\run = \run' \revent$ then for some $\qstate \in \stateset$, $\qstate\movesto{\outopof\thread\xvar\areg}\pqstate$ is in the $\transitionset$ of $\rmachine$.
 Thus clearly, from Rule $\bluecircled{1}$, $\tuple{\areg}\in \regs(\qstate)$.
  Since $\qstate$ is reachable by $\initstate$ through $\run$,
  we need to see what sort of edge would be the last one on the path
  between $\initstate$ and $\qstate$.
  We check the rules of Figure \ref{fig:ghost_reads} and
  and see that
  the only way that the update $\tuple{\areg}\in \regs(\initstate)$ does not take place,
  is when we are \textit{missing} a rule in Figure \ref{fig:ghost_reads}
  for a certain type of edge.
The only type of edge which does not propagate $\tuple{\areg}$ in
 Figure \ref{fig:ghost_reads} is an input edge on $\areg$,
  or a copy from a different register $\breg$ to $\areg$.
Of these two we know that the first is
 impossible due to the premise of the statement we are proving.
 In the second case,
 we are left with a shorter run, and the same statement to prove,
  but for the new
 register $\breg$.
Since the premise of the statement forbids from any input writing on the register that is the current register at any point we know that some register will end up being associated to $\initstate$, and thus trigger the rule $\bluecircled{1}$
\end{proof}

The complxity of Algorithm \ref{algo:ghost_reads_pseudo} is calculated as follows:
The algorithm will use linear time to calculate the set of initial tuples
(lines $3$ and $4$). The while loop of line $5$ is executed while there are
still new tuples added each iteration. The \textbf{for}-loop within the while
loop takes $no_of_tuples \times n$, which is $\bigO{n}$. In the worst case the
while loop will add at most one register to one state at each such iteration,
so in the worst case it is executed $\bigO{n^2}$ times. Therefore the final
complexity is $\bigO{n^3}$, which however can be severely improved
(as our tool does) to only trigger the while loop when a new touple is actually added,
without searching from scratch for new rules to add.

\subsection{Mismatched Variables}\label{app:mismatch_variables}

\label{sec:mismatched_vars}

In this part we tackle the problem of a register machine $\rmachine$ outputting the value of a register in a variable other than the variable that the register was written for.
The algorithm is very similar to that for ghost reads, only now we need to keep track of extra information regarding the variable whose value was stored in a register.
For this violation to occur we still need a output event to take place, and thus we will initialize the search from output edges in $\rmachine$. Figure \ref{fig:mismatched_variables} summarizes the update rules for the data structure used by the algorithm.
\begin{figure}[t]
\[
\begin{array}{c}
\text{Initialization and termination} \\
\begin{array}{c}
\scalebox{0.6}{
\onelabelinfer{\iqstate1\mkrelof{\rtype\land\thread\land\areg\land\xvar}\iqstate2
}{
  \regs(\qstate_1)\oplus\tuple{\areg,\xvar}
}{$\largebluecircled{1}$}}
\end{array}
\begin{array}{c}
\scalebox{0.6}{
\onelabelinfer{\iqstate1\mkrelof{\wtype\land\thread\land\areg\land\yvar}\iqstate2 \irspace
 \tuple{\areg,x}\in \regs(\qstate_2)
}{
  \failconf\oplus\qstate_1
}{$\largebluecircled{2}$}}
\end{array}
\\
\text{Copied registers} \\
\begin{array}{c}
\scalebox{0.6}{
\onelabelinfer{\iqstate1\mkrelof{\areg\assigned\breg}\iqstate2 \irspace
 \tuple{\breg,x}\in \regs(\qstate_2)
}{
  \regs(\qstate_1)\oplus\tuple{\breg,x}
}{$\largebluecircled{3}$}}
\end{array}
\begin{array}{c}
\scalebox{0.6}{
\onelabelinfer{\iqstate1\mkrelof{\areg\assigned\breg}\iqstate2 \irspace
 \tuple{\areg,x}\in \regs(\qstate_2)
}{
  \regs(\qstate_1)\oplus\tuple{\breg,x}
}{$\largebluecircled{4}$}}
\end{array}
\\
\text{Transparency} \\
\begin{array}{c}
\scalebox{0.6}{
\onelabelinfer{\iqstate1\mkrelof{\rtype\land\thread}\iqstate2 \irspace
 \tuple{-,-}\in \regs(\qstate_2)
}{
  \regs(\qstate_1)\oplus\tuple{-,-}
}{$\largebluecircled{5}$}}
\end{array}
\begin{array}{c}
\scalebox{0.6}{
\onelabelinfer{\iqstate1\mkrelof{\wtype\land\thread\land\areg}\iqstate2 \irspace
 \tuple{\breg,-}\in \regs(\qstate_2)
}{
  \regs(\qstate_1)\oplus\tuple{\breg,-}
}{$\largebluecircled{6}$}}
\end{array}
\\
\begin{array}{c}
\scalebox{0.6}{
\onelabelinfer{\iqstate1\mkrelof{\areg\assigned\breg}\iqstate2 \irspace
 \tuple{\creg,-}\in \regs(\qstate_2)
}{
  \regs(\qstate_1)\oplus\tuple{\creg,-}
}{$\largebluecircled{7}$}}
\end{array}
\end{array}
\]
\caption{The inference rules for updating the $\regs$ data structure , for each state $\qstate \in \stateset$ of a register machine $\rmachine$ during the mismatched variable reads module. We assume that $\xvar\neq\yvar$ and $\breg\neq\areg$. A $-$ in a tuple means that that particular attribute is of no importance to the rule, but will be carried over in the update.}
\label{fig:mismatched_variables}
\end{figure}
Our algorithm works as follows:

\begin{itemize}
\item
It uses a data structure $\funtype{\regs}{\stateset}{2^{\regset*\varset}}$.
If $\tuple{\areg,\xvar}\in \regs(\qstate)$ then there is a run $\run$ of $\rmachine$ from $\qstate$, where whatever was stored in $\areg$ at $\qstate$ will be outputted to variable $\xvar$.

\item The algorithm first isolates all output edges $E_{out}$ of the register and then searches backwards through the edges or $\rmachine$. We update the data structure as Figure \ref{fig:mismatched_variables} indicates.

\item If at any point a state marked with $\tuple{\areg,\xvar}$ can be reached by a $\qstate\movesto{\inopof\thread\areg\yvar -}$ transition , it means that the value of $\yvar$ will be stored in $\areg$, but later on a run can continue and output it on $\xvar$. We assume the register machine is able to reach all states, which means this is a violation.
\end{itemize}

Figure \ref{fig:ghost_read_violation} can also be used as
an example of a register machine that will flag an $\failconf$ configuration
 in this module. The run $
\run = \qstate_0\movesto{\inopof{\theta_1}{x}{a_1}{-}}\qstate_1\movesto{\outopof{\theta_2}{x}{a_1}}\qstate_2\movesto{\inopof{\theta_2}{y}{a_2}{-}}\qstate_0\movesto{\inopof{\theta_1}{x}{a_1}{-}}\qstate_1\movesto{\outopof{\theta_2}{x}{a_1}}\qstate_2\movesto{\outopof{\theta_1}{x}{a_2}}\qstate_3$ outputs the value of $\areg_2$ on $\xvar$, even though during $\run$ it inputted in $\areg_2$ on $\yvar$.

We now proceed similarly with the algorithm for Question \ref{q:mismatched_variables}. The proof is very similar to the previous one. I.e.  we need to show that 1) our rules indeed mark and propagate register-variable pairs correctly, and 2) the rules mark all possible dangerous register-variable pairs per state.

\begin{proof}[Proof: Algorithm of Figure \ref{fig:mismatched_variables}]
The proof follows the same structure as above. We can use induction on the number of steps to prove that associating a variable-register tuple to a state preserves the premise, and for the reverse it suffices to check that no valid tuple is stopped from propagating backwards when the relevant path exists.
\end{proof}
From now on we assume that the algorithms of Figures \ref{fig:ghost_reads}
an \ref{fig:mismatched_variables} have been run in advance,
and thus the register machines we are checking for cyclic
$\mkrelof{\wtype\land\varattrof\wevent}\cdot\hbrel\cdot\wevent\cdot\hbrel\cdot\invmkrelof\rfrel$
relations have successfully passed those checks.
Below we present our algorithm for checkig the third condition.

\subsection{Exposed Read Violations}
\label{sec:hidden_reads}
Assume we are given a register machine $\rmachine=\rmachinetuple$ operating on a set $\thread$ of threads, a set $\varset$ of variables, a set $\regset$ of registers, and a set $\valset$ of data values.
We present an algorithm that checks whether all runs of $\rmachine$
respect the third condition of Proposition \ref{prop:algo_RA_three_parts}.
\begin{figure}
\begin{center}
  \text{The rules for $\wrasem$. All rules apply only for $a\neq b$, and $x \neq y$.}
\end{center}
\[
\begin{array}{lll}
\scalebox{0.55}{
\onelabelinfer{\iqstate1\mkrelof{\rtype\land\thread\land\xvar\land\areg}\iqstate2
}{
  \visibleof{\iqstate1}\areg\thread\xvar{}
}{$\largebluecircled{1}$}}
&
\scalebox{0.55}{ \twolabelinfer{
  \iqstate0\mkrelof{\rtype\land\thread\land\yvar\land\breg}\iqstate1
  }{
  \statusof{\iqstate1}\areg\thread\xvar{}
  }{
  \statusof{\iqstate0}\areg\breg\xvar{}
}{$\largebluecircled{2}$} }
&
\scalebox{0.55}{ \twolabelinfer{
  \iqstate0\mkrelof{\wtype\land\thread\land\yvar\land\breg}\iqstate1
  }{
  \statusof{\iqstate1}\areg\breg\xvar{}
  }{
  \statusof{\iqstate0}\areg\thread\xvar{}
}{$\largebluecircled{3}$} }
\\ \\
\scalebox{0.55}{\twolabelinfer{
  \iqstate0\mkrelof{\wtype\land\thread\land\xvar\land\breg}\iqstate1
  }{
  \statusof{\iqstate1}\areg{\thread \text{ or } \breg}\xvar{}
  }{
  \hiddenof{\iqstate0}\areg\thread\xvar{}
}{$\largebluecircled{4}$} }
&
\scalebox{0.55}{\twolabelinfer{
  \iqstate0\mkrelof{\rtype\land\thread\land\xvar\land\breg}\iqstate1
  }{
  \statusof{\iqstate1}\areg\thread\xvar{}
  }{
  \hiddenof{\iqstate0}\areg\breg\xvar{}
}{$\largebluecircled{5}$} }
&
\scalebox{0.55}{\twolabelinfer{
  \iqstate0\mkrelof{\rtype\land\thread\land\xvar\land{(\areg \text{ or }\breg)}}\iqstate1
  }{
  \hiddenof{\iqstate1}{\areg}\thread\xvar{}
  }{
  \hiddenof{\iqstate0}{\areg}{\areg}\xvar{}
}{$\largebluecircled{6}$} }
\end{array}
\]

\text{\hspace{1.7cm} Detecting failures. \hspace{2.7cm} Handling copy events
}
  \[
\begin{array}{ccc}
\scalebox{0.55}{
\twolabelinfer{
  \iqstate0\mkrelof{\wtype\land\thread\land\xvar\land\areg}\iqstate1
  }{
  \hiddenof{\iqstate1}\areg{\thread  \text{ or } \areg}{\xvar}{}
  }{
  \failconf
}{$\largebluecircled{7}$} }
&
\quad
&
 \scalebox{0.55}{
 \twolabelinfer{
   {\iqstate0}\mkrelof{\ceventof{\areg}{\breg}}{\iqstate1}
   }{
   \statusof{\iqstate1}{reg}{data}{\xvar}{}
   }{
   \statusof{\iqstate0}{reg\{\sfrac{b}{a}\}}{data\{\sfrac{b}{a}\}}\xvar{}
   }
   {$\largebluecircled{9}$} }
\end{array}
\]
\begin{center}
  \text{Transparency rules propagate \textbf{ANY} tuples to states when the occurring event}
\text{ does not affect the stored information.}
\end{center}
\[
\begin{array}{cccc}
  \qquad\quad
  &
\scalebox{0.55}{
 \twolabelinfer{
   {\iqstate0}\mkrelof{\wtype\wedge \creg}{\iqstate1}
   }{
   \statusof{\iqstate1}{\areg}{\thread  \text{ or } \breg}{\xvar}{}
   }{
   \statusof{\iqstate0}{\areg}{\thread \text{ or } \breg}{\xvar}{}
   }
   {$\largebluecircled{8}$} }
   &
   \quad
   &
 \scalebox{0.55}{
   \twolabelinfer{
   {\iqstate0}\mkrelof{\rtype}{\iqstate1}
   }{
   \statusof{\iqstate1}{\areg}{\thread  \text{ or } \breg}{\xvar}{}
   }{
   \statusof{\iqstate0}{\areg}{\thread \text{ or } \breg}{\xvar}{}
   }
   {$\largebluecircled{8'}$} }
\end{array}
\]

\caption{The rules for updating the $\hidden$, and $\visible$ data structures. $\status$ means either $\hidden$ or $\visible$. These rules are sufficient to detect violations of $\wrasem$, and for tracking $\porel \cup \rfrel$. When a rule is stated with an ``or'' description it stands for two rules, one for each version of this clause. }
\label{fig:rules_wra}
\end{figure}
\begin{figure}
  \begin{algorithm}[H]
  \KwIn{$\rmachine$}
   $rule\_updated := \true $\;
   tuples := $\emptyset$\;
   \For{$read\_edge \in \rmachine$}{
   tuples := tuples~$\cup~Rule1.(read\_edge)$\;
   }
   \While{$rule\_updated$}{
    $rule\_updated := \false$\;
    \For{ Rule $ \in$  Fig. \ref{fig:rules_wra}}{
      \If{Rule.$condition() \in tuples$}{
        \eIf{Rule.update() ==~$\failconf$}{
        \KwRet{$\failconf$}\;
        }{
          $tuples = tuples~\cup$ Rule.$update()$\;
          $rule\_updated := \true$\;
        }
      }
    }
  }
   \caption{The algorithm for $\wrasem$}\label{algo:wra_pseudo}
  \end{algorithm}
  \end{figure}
%
Algorithm \ref{algo:wra_pseudo} assumes that $\rmachine$ has successfully passed the previous two stages of the verification (namely ghost reads and mismatched variables).
The key contributor to the complexity is the size of the
 largest data structure necessary.
For $\wraconsprob$ we need data structures of size
$\bigO{|\stateset|*|\varset|*|\regset|*|(\threadset \cup \regset)|}$,
 which simplifies to $\bigO{n^{4}}$.
These data structures are used to store {\it summaries} of possible runs
in $\rmachine$.
We only keep information for runs that might cause violations
 as the ones described
in Proposition \ref{prop:algo_RA_three_parts}.
The addition of information (in the form of state-summary tuples) to our data structures is monotonic,
 meaning once
we have discovered a path from some state $\qstate$ exists,
 this information is never removed.
Thus the data structure size bounds the number of iterations
 of our main loop (line $5$, Algorithm \ref{algo:wra_pseudo}).
 For each iteration, the algorithm checks if an existing possible violation
 path can be combined with more transitions of $\rmachine$ (a search which is linear
 to the size of the machine), and creates more tuples to new states.
 The conditions for guranteeing existence of new (relevant)
  paths is characterized by the rules of Fig. \ref{fig:rules_wra}, which produce new
  tuples (see lines $4$ and $13$).
  If no update takes place in one iteration, we are guaranteed we have arrived at a fixed-point
  and the proceedure stops (lines $1$, $6$, and $14$).
Our data structures are:


\begin{itemize}
\item
$\funtype\visible{\stateset}{2^{\varset \times(\threadset\cup\regset)\times\regset}}$.
If $\visibleof{\qstate}\areg\thread\xvar{}$
then there is a run $\run$ of $\rmachine$ starting at $\qstate$
in which $\thread$ has a  $\tclosureof{\mkrelof{\porel\cup\rfrel}}$
path to a read event $\revent$ on $\xvar$, whose value is stored
in register $\areg$ when $\run$ is in $\qstate$
(for example rule $\largebluecircled{1}$ in Fig. \ref{fig:rules_wra}).
A tuple $\visibleof{\qstate}\areg\breg\xvar{}$ is created when,
along an existing path to such an $\revent$, 
 we encounter a new read event $\revent'$, which reads from the value of
 register $\breg$ on $\thread'$ (see rule $\largebluecircled{2}$).
The tuple stands for an \textit{expectation} of a tuple
$\visibleof{\qstate'}\areg{\thread'}\xvar{}$,
as, when an event that inputs the value for $\revent'$, we
 should
get a new $\tclosureof{\mkrelof{\porel\cup\rfrel}}$ path, and thus
we should get the tuple $\visibleof{\qstate'}\areg{\thread'}\xvar{}$
(see rule $\largebluecircled{3}$ ).
We refer to  Example \ref{exp:tuple_addition} for a demonstration of this procedure. 

%

\item
$\funtype\hidden{\stateset}{2^{\varset \times(\threadset\cup\regset)\times\regset}}$.
These tuples are essentially an extension of the $\visible$ tuples,
 with the addition that the paths that were described above can now
 refer to
  $\tclosureof{\mkrelof{\porel\cup\rfrel}}\wevent_x\tclosureof{\mkrelof{\porel\cup\rfrel}}$
   paths instead (for all $\xvar \in \varset$).

%
%

\end{itemize}

\begin{figure}
	\begin{subfigure}[t]{0.5\textwidth}
	\centering
	\scalebox{0.8}{
    \begin{tikzpicture}[every text node part/.style={align=center}]
        \node[state, initial] (q1) at (0, 0) {$q_1$};
        \node[state] (q2) at (3.5, 0) {$q_2$};

        \draw[-stealth] (q1) edge[loop above] node[above,midway] {$\inopof{\theta}{x}{a}{}$ \\ $\inopof{\theta}{x}{b}{}$  \\
         $\inopof{\theta}{y}{c}{}$ } (q1);
        \draw[-stealth] (q1) to[bend left] node[above,midway]
        {$\outopof{\phi}{y}{c}$} (q2);
		\draw[-stealth] (q2) edge[loop above] node[above,midway] {$\outopof{\phi}{x}{a}{}$ } (q2);
    \end{tikzpicture}}
    \caption{Register machine $\rmachine$.}
    \label{fig:rmachine_not_wra}
  \end{subfigure}
  \begin{subfigure}[t]{0.5\textwidth}
  \scalebox{0.8}{
  \centering
\begin{tikzpicture}
          \node (e1) at (0,3.3) {$e_1: \eventof\wtype{\theta}{x}{1}$};
          \node (e2) at (0,1.6) {$e_2: \eventof\wtype{\theta}{x}{2}$};
          \node (e3) at (0,0) {$ e_3: \eventof\wtype{\theta}{y}{1}$};
          \node (e4) at (4,2.5) {$e_4: \eventof\rtype{\phi}{y}{1}$};
          \node (e5) at (4,0.8) {$e_5: \eventof\rtype{\phi}{x}{1}$};

          \draw[-stealth, \pocolor] (e1) -- (e2) node[midway, fill=white] {$\porel$};
          \draw[-stealth, \pocolor] (e2) -- (e3) node[midway, fill=white] {$\porel$};
          \draw[-stealth, \pocolor] (e4) -- (e5) node[midway, fill=white] {$\porel$};
          \draw[-stealth, dashed]  (e1) -- (e5) node[midway, fill=white] {$\rfrel$};
          \draw[-stealth, dashed]  (e3) -- (e4) node[midway, fill=white] {$\rfrel$};
          \draw[-stealth ,dotted] (e2) to [bend left] (e1) node[yshift=-0.3cm,xshift=-0.6cm,fill=white] {$\corel_x$};
      \end{tikzpicture} }
      \caption{A run $\run$ of $\rmachine$.}
      \label{fig:run_not_wra}
   \end{subfigure}
\caption{A machine that produces an execution graph violating the $\wrasem$ semantics. }
\label{fig:not_wra}
\end{figure}

\begin{proposition}\label{prop:modules_per_memory_model}
  For a register machine $\rmachine$,  $\rmachine\models\wrasem$ iff running Algorithm \ref{algo:wra_pseudo} 
   on $\rmachine$
   flags no $\failconf$ configurations.
  \end{proposition}

We will give the detalied proof of this statement
after we discuss an example application of the algorithm.

\begin{example}\label{exp:tuple_addition}
   Consider the register machine $\rmachine$ shown in Fig. \ref{fig:rmachine_not_wra}, and its run $\run$ shown in Fig. \ref{fig:run_not_wra}.
   This execution graph does not respect $\wrasem$, as reversing the $\rfrel$ edge from $e_1$ to $e_5$, which means that the rules of Figure \ref{fig:rules_wra} should flag $\failconf$. We show the application of the rules for checking $\wrasem$, as stated by Proposition \ref{prop:modules_per_memory_model}.
\begin{enumerate}
\item
    The edge $ \qstate_2\movesto{\outopof{\phi}{\xvar}{\areg}}\qstate_2$ matches rule
     $\bluecircled{1}$, which creates the tuple \\
      $\visibleof{\iqstate2}{\areg}{\phi}\xvar{}$.
    In this stage this tuple keeps track of the possible path that starts
     in $\qstate_2$, and can perform a read on $\xvar$ (the event $\event_5$)
      that is
     $\tclosureof{\mkrelof{\porel\cup\rfrel}}$ connected to events
      that occur earlier in $\phi$.
\item
    $\qstate_1\movesto{\outopof{\phi}{\yvar}{\creg}}\qstate_2$
     together with $\visibleof{\iqstate2}{\areg}{\phi}\xvar{}$
      matches rule $\bluecircled{2}$, which in turn adds the tuple
   $\visibleof{\iqstate1}{\areg}{\creg}\xvar{}$.
    This event is relevant to the $\tclosureof{\mkrelof{\porel\cup\rfrel}}$
     path in the execution graph that we are keeping track of.
      Namely, any thread that will input the value that will be
       stored in $\creg$ when reaching $\qstate_1$ will gain a
        $\tclosureof{\mkrelof{\porel\cup\rfrel}}$ path to $\event_5$.
\item
    $\qstate_1\movesto{\inopof{\theta}{\yvar}{\creg}{}}\qstate_1$
    together with $\visibleof{\iqstate1}{\areg}{\creg}\xvar{}$ match
    rule $\bluecircled{3}$, adding the tuple
    $\visibleof{\iqstate1}{\areg}{\thread}\xvar{}$.
    Here, we see that indeed $\thread$ performs the input that was marked
     by the tuple in $\qstate_1$.
    Thus, as can be confirmed at the execution graph in Fig.
     \ref{fig:not_wra}, an $\rfrel$ edge now relates $\thread$
     to $\event_5$.
\item
    $\qstate_1\movesto{\inopof{\theta}{x}{\breg}{}}\qstate_1$
    together with $\visibleof{\iqstate1}{\areg}{\thread}\xvar{}$
     matches rule $\bluecircled{4}$, which adds the tuple
    $\hiddenof{\iqstate1}{\areg}{\thread}\xvar{}$.
    This update now initializes a $\hidden$-type tuple.
    This is because now not only we have confirmed the
     $\tclosureof{\mkrelof{\porel\cup\rfrel}}$ path to $\event_5$,
      but also we have observed an input on the thread $\thread$
       for $\xvar$, which means that all other prior input events
       that will take place on it, should be ``hidden'' from $\event_5$.
\item
    $\qstate_0\movesto{\inopof{\theta}{\xvar}{\areg}{}}\qstate_1$
    together with   $\hiddenof{\iqstate1}{\areg}{\thread}\xvar{}$
    matches rule $\bluecircled{7}$, and produces the
    $\failconf$ flag.
\end{enumerate}
\end{example}
Example \ref{exp:tuple_addition} only demonstrates the functionality of rules $\bluecircled1$ ,$\bluecircled2$, $\bluecircled3$ ,$\bluecircled4$, and $\bluecircled{7}$.
The rest of the rules capture different ways of how paths can be propagataed to new states
in a register machine. The exact function of each rule will be
demonstrated in the following proof of correcnes.

Before proving the final algorithm, we will prove the following lemma:

\begin{lemma}
For register machine $\rmachine$ and run $\run\in\runsetof\rmachine$, with $  \tuple{\initconf,\egraph_{\emptyset}} \cmovesto{\run}  \tuple{\conf_1,\egraph} $, such that $\mkrelof{\wtype}\cdot\hbrel\cdot\wevent\cdot\hbrel\cdot\invmkrelof\rfrel$ is cyclic in $\egraph$, then there exist $\run_1,\run_2$, such that $\run = \run_1 \cdot \run_2$ and $\run_1 = \run_0\cdot \revent$, and for any prefix of $\run_0$, the corresponding execution graph is acyclic on $\mkrelof{\wtype}\cdot\hbrel\cdot\wevent\cdot\hbrel\cdot\invmkrelof\rfrel$ for all $\xvar$, and $  \tuple{\initconf,\egraph_{\emptyset}} \cmovesto{\run_1}  \tuple{\conf',\egraph_1} $, with $\egraph_1 \not\models \wrasem$.
\end{lemma}

\begin{proof}
To prove this we observe the rules of Figure \ref{add:event:fig}. We clearly see that if a prefix $\run_t$ of the run $\run$ the corresponding execution graph $\egraph_t$ contains no cycles and the next event in $\run$ is a write or a copy, then the updates that will take place on the graph do not create any cycles.
Thus the only way to create one is by adding an appropriate read event.
Since we start from the empty execution graph $\egraph_0$, witch of course does not contain cycles, as it contains no events and all the relations are empty, we just define as $\run_1$ the first prefix of $\run$ for which a cycle is formed (since one has to exist in order for $\egraph$ to contain one. Thus the promise holds for $\run_1$.
\end{proof}
In what follows when studying a run $\run$ for which the corresponding execution graph contains a cycle for $\mkrelof{\wtype}\cdot\hbrel\cdot\wevent\cdot\hbrel\cdot\invmkrelof\rfrel$ for some $\xvar$, we will assume that $\run = \run_0\cdot \revent$, and the execution graph corresponding to $\run_0$ has no cycles.
For such a violating run we state the following:

\begin{proposition}\label{prop:spliting_of_abstracted_run}
Assume a register machine $\rmachine$ and $\rmachine\not\models\wrasem$. Then, there exists a differentiated run $\run \in \rmachine$, where
\begin{itemize}
\item  $\run = \run_1\cdot\wevent\cdot\run_2\cdot \revent$, $\wevent = \eventof\wtype{\thread_1}\xvar\val$, and $\revent = \eventof\rtype{\thread_2}\xvar\val$ (possibly $\thread_1 = \thread_2$), and

\item  $\tuple{\initconf,\egraph_{\emptyset}} \cmovesto{\run}  \tuple{\conf_1,\egraph} $,
where $\egraph = \tuple{\events,\porel,\rfrel,\corel}$, and $\conf_1 = \tuple{\qstate,\regvalmapping}$, and
      \item $\mkrelof{\wtype}\cdot\hbrel\cdot\wevent\cdot\hbrel\cdot\invmkrelof\rfrel$ contains a cycle
\end{itemize}

\end{proposition}

Intuitively, this states that two things need to happen in order for a cycle to occur.
\begin{itemize}
\item A $\wevent$~needs to occur in some thread $\thread_1$.
\item This $\wevent$~then must not available for a thread $\thread_2$ to read due to other events that have occurred already.
\item The value of $\wevent$~needs to be outputted for a read event $\revent$ after said ``hiding'' on $\thread_2$.
\end{itemize}

We now focus on the second condition of the above list, i.e. we will try to see the exact situation upon which a given $\wevent$ is not available to be read by a thread, because a cycle would be created. We refer to  Figure \ref{add:event:fig}, which shows us what kind of updates would take place in a graph where this forbidden read event takes place. In this case $\pevent = \wevent$. There we see:

\begin{itemize}
\item The updates to the $\porel$ and $\rfrel$ relations cannot create a cycle, as the previous execution graph is acyclic (not that $\run$ is assumed to be minimal), and the added edges add one-directional edges between the existing graph and the added new node for $\revent$.
\item The cycle that is created upon updating the execution graph must have come from the update rule for $\corel$, which means that there must be at least on more $\wevent'$ in $\run$, on variable $\xvar$ that satisfies the relevant update condition.
\end{itemize}

We therefore obtain one extra piece of information for this minimal, differentiated, violating run $\run$, and this is that it contains another write event $\wevent'$ on variable $\xvar$. We now focus on determining which of the $\wevent$ and $\wevent'$ precedes the other in $\run$.

We assume $\wevent'$ is earlier in $\run$.
 Since $\run$ is minimal and thus there are no cycles in $\mkrelof{\wtype}\cdot\hbrel\cdot\wevent\cdot\hbrel\cdot\invmkrelof\rfrel$ until the event $\revent$ occurs.
 Therefore, we have that the addition of the $(\wevent', \wevent)$ that the the update rule for $\corel$ dictates would not cause a cycle (as all old and added edges are pointing in a consistent time-order towards the $\revent$.

 Thus it must be the case that $\run = \run_1\cdot\wevent\cdot\run_2\cdot \revent$, and $\run_2 = \run_{\beta}\cdot \wevent' \cdot \run_{\alpha}$.
  The naming here aims to relate the $\hidden$ tuples that will be produced from Fig.  \ref{fig:rules_wra} to the $\run_{\beta}$ part of the run, and $\visible$ to the $\run_{\alpha}$ part.
  Based on this partition of $\run$ and the fact that the $\revent$ event triggers an update of $\corel$ that cases a cycle, we state the following:

\begin{proposition}\label{prop:splitting_violating_run}
Assume a register machine $\rmachine$ and $\rmachine\not\models\wrasem$. Then, there exists a minimal differentiated run $\run \in \rmachine$, where
\begin{itemize}
\item  $\run = \run_1\cdot\wevent\cdot\run_{\beta}\cdot \wevent' \cdot\run_{\alpha}\cdot \revent$, $\wevent = \eventof\wtype{\thread_1}\xvar\val$, and $\revent = \eventof\rtype{\thread_2}\xvar\val$, $\wevent' = \eventof\wtype{\thread_3}\xvar-$ (possibly $\thread_1 = \thread_2 = \thread_3$), and
\item $\wevent' \mkrelof{\porel \cup \rfrel} \revent $, and
\item $\revent \mkrelof{\invmkrelof\rfrel} \wevent \mkrelof{\porel \cup \rfrel }\wevent'$
\end{itemize}
\end{proposition}

We now prove a short lemma regarding where the value of $\revent$ is  stored at.

\begin{lemma}\label{lem:rules_store_current_register}
For any state $\qstate$ of the register machine $\rmachine$, that was accessed during $\run_{\beta}\cdot \wevent'\cdot \run_{\alpha}$, we have that there is some tuple in one of the $\visible \text{ or }  \hidden$ data structures of $\tuple{\qstate,\pqstate}$, with contents $\tuple{-,\creg}$ where the value of $\revent$ was in register $\creg$ when $\qstate$ was accessed.
\end{lemma}
We prove this for the rules of Figure \ref{fig:rules_wra}.
\begin{proof}
The proof of this lemma is relatively straightforward. We see that the value of $\revent  = \outopof\thread\xvar\areg$  is initially at register  $\areg$. The only reason this will change while traversing $\run$ backwards is either a copy command, or a input on this (or the current) register.

In the case of a copy command we see that indeed the register identifying the $\visible$ and $\hidden$ tuples that the state is assigned to are updated accordingly through rule $\bluecircled{9}$, which handles all possible ways that the value of $\areg$ would have originated in a different register.

In the case of an input command on the current register $\creg$  we see that now the value of that is outputted in $\outopof\thread\xvar\areg$ is no longer in \textit{any} register, and we observe that accordingly no rules propagate a tuple to the backwards reachable state.
\end{proof}
In essence the above lemma captures that the rules of  Figure \ref{fig:rules_wra} maintain at least as much information as those of  Figure \ref{fig:mismatched_variables}.
We highlight that in all the above discussions we have $\revent$ must reading from $\wevent$.
This should be clear as it is the main assumption for the existence of the violating run and the relevant run fragments.
Therefore we can safely assume, combining this with Lemma \ref{lem:rules_store_current_register} that throughout $\run_{\beta}\cdot \wevent' \cdot\run_{\alpha}\cdot$ any activated transition is never writing (by input) a new value to the current register stored at the tuples of the state reached by the transition.

\begin{corollary}\label{cor:run_does_not_overwrite}
For the above violating run and partition:
\begin{itemize}
\item  $\wevent$ is the latest (rightmost) input on the register that the value of $\revent$ is stored at during that event of the run.
\item All input events after $\wevent$ in $\run$ target registers different than the one where the value of $\revent$ is stored at during that event of the run.
\end{itemize}

\end{corollary}


We make the following claims:

\begin{lemma}\label{lem:copy_edges_are_sound}
Let $\run$ be the run of Proposition \ref{prop:splitting_violating_run}, and $\qstate\movesto{
\op}\pqstate$ be a copy $\breg := \creg$ edge,
 accessed during the $\run_{\beta}\cdot \wevent' \cdot\run_{\alpha}$
  part of $\run$.
Then, for any tuple $T\in \status(\pqstate,\pqstate')$ implies there
 exists another tuple
 $T'\in \status(\qstate,\pqstate)$, with $T'$ being the weakest
 register pre-condition with respect to $\breg := \creg$
  (hence, all ocurences of 
 $\breg$ in $T$ are replaced with $creg$ in $T'$). 
\end{lemma}

\begin{proof}
This proof is relatively straightforward. Note that due to Lemma \ref{lem:rules_store_current_register}, the registers $\creg$ and $\breg$ are guaranteed to be exactly the registers where the value that will be outputted in $\revent$ is currently stored at.
We proceed with case analysis on the form of $\op$.

\begin{itemize}
\item if $\op$ overwrites the value of $\breg$,  with $\breg$ the register holding the value of $\revent$ in $\run$. We see the update that will take place in this case is rule $\bluecircled{9}$.
The update to the current register for $\revent$ is correct and the premise holds.
\item if $\op$ overwrites the value of $\areg$, where $\areg$ held the value of some other register that will be outputted later on on the run $\run$ towards $\revent$, then similarly, the activated rule is $\bluecircled{9}$, which update only the relevant tuples accordingly, but not the register for the value of $\revent$, and the statement holds.
\item The copy edge overwrites the value of a register that is neither part of the tuple associated to $\pqstate$, or the register where the value of $\revent$ is stored at.

We see that the only rule applied then is $\bluecircled{9}$, which correctly does not update any of the registers, but does assign the claimed tuples to $\qstate$.
\end{itemize}
\end{proof}

We now have that all copy edges have been proven to soundly propagate whatever information was stored in some state regarding registers, to new states.
We combine the above lemma and Corollary \ref{cor:run_does_not_overwrite} to obtain:
\begin{corollary}\label{cor:register_tuple_not_dangerous_value}
Let $\run$ be the run of Proposition \ref{prop:splitting_violating_run}, and $\qstate$ be a state accessed during the $\run_{\beta}\cdot \wevent' \cdot\run_{\alpha}$ part of $\run$.
If $\areg$ is the secondary register in of any $\status$ tuple, then $\areg$  is not the register where the value of $\revent$ is currently stored at.
\end{corollary}

We proceed with some more general statements:
\begin{lemma}\label{lem:marked_registers_will_be_read}
Let $\run$ be the run of Proposition \ref{prop:splitting_violating_run}, and $\qstate$ be a state accessed during the $\run_{\beta}\cdot \wevent' \cdot\run_{\alpha}$ part of $\run$.
If $\areg$ is the secondary register in of any $\status$ tuple, then it
contains a value that will be outputted at some event $\event$ in $\run$ that follows $\qstate$.

\end{lemma}

\begin{proof}
We prove this with a relatively straightforward induction on the number register-register tuples that have been added. Let $t_{\areg}$ denote such a tuple.

\underline{Base case:} $\areg$ is outputted immediately.

This can only be true if the tuple $t_a$ is created due to an output edge.
This can only be true for rules $\bluecircled{2}$, $\bluecircled{5}$, and $\bluecircled{6}$.
For all of those we see that indeed the claim holds, and it is true that $\areg$ is storing the value of a register that will be outputted later in the run.

\underline{Inductive hypothesis:} The claim holds for all existing $t_{\areg}$ tuples.

\underline{Inductive Step:} Assume a new tuple $t_{\areg}$.
This tuple could have been created due to rules $\bluecircled{2}$, $\bluecircled{5}$, $\bluecircled{6}$, and $\bluecircled{8}$,  $\bluecircled{8'}$ and $ \bluecircled{9}$.
For the first three we know that the claim holds immediately.
For the copy rule(s), in order for one to be activated, we see that a tuple $t_{\breg}$ must also have been associated to the relevant set at the next state accessed in $\run$, and from Lemma \ref{lem:copy_edges_are_sound} we get that $\breg$ must be the register where the value of $\areg$ was stored at when reaching $\qstate$.

Thus, from inductive hypothesis, $\breg$ contains a value that will be outputted later in $\run$.
Consequently, after traversing the copy edge backwards the value of $\breg$ is now in $\areg$ and therefore the claim holds.

Regarding rule $\bluecircled{9}$, we have an identical argument which is thus omitted.
\end{proof}



 To generalize even more the above we have:

\begin{lemma}\label{lem:tuples_imply_paths}
Let $\run$ be the run of Proposition \ref{prop:splitting_violating_run}, and $\qstate\movesto{\opof\thread---}\pqstate$ be an input or output edge that is accessed during the $\run_{\beta}\cdot \wevent' \cdot\run_{\alpha}$ part of $\run$,
and the value that will be outputted on $\revent$ is  stored in $\creg$.
Assume that $\pqstate$ is followed by several events which belong in $\mkrelof{\porel \cup \rfrel}^* \revent$ paths in the execution graph $\egraph$ coresponding to $\run$.
 Then the following hold:

 \begin{itemize}
 \item $\statusof{\qstate}\areg\thread\xvar{\pqstate}$ is added if $\thread$ occurs in some existing $ \mkrelof{\porel \cup \rfrel}^* \revent$  path or
 if $\op$ is an input on a new thread, but the value that is inputted has been later on read by an event in one of the threads in one $ \mkrelof{\porel \cup \rfrel}^* \revent$ path.
 \item $\statusof{\qstate}\areg\breg\xvar{\pqstate}$ is added if the set of threads that can reach the event $\revent$ though a
  $ \mkrelof{\porel \cup \rfrel}^* \revent$ path does not chance from $\pstate$ to $\qstate$, and the value currently in $\breg$ will be outputted later on in one of the threads that are still in a  $ \mkrelof{\porel \cup \rfrel}^* \revent$ path.
 \end{itemize}

\end{lemma}

 \begin{proof}
We prove this by induction.

 \underline{Base case:} Rule $\bluecircled{1}$, and $\event = \revent$. In this case the statement holds trivially.

\underline{Hypothesis} for all existing  $\statusof{\qstate}\areg\thread\xvar{\pqstate}$ tuples the claim holds.

\underline{Inductive step:}
We inspect a new added tuple.

Say it is of type thread-register.
We look up the rules of Figure \ref{fig:rules_wra} to see where this update could have come from. We quickly see that of all the rules that could have added this tuple, the only ones where the inductive hypothesis does not hold immediately is : $\bluecircled{3}$, and $\bluecircled{4}$. The case analysis for these rules is very similar so we only do one.
\begin{itemize}
\item Rule $\bluecircled{3}$, and $\event = \inopof\thread\breg\yvar-$.
For this to be the case, we have that there must have also existed tuple,  $\statusof{\pqstate}\creg\breg\xvar{\pqstate'}$, and from Lemma \ref{lem:marked_registers_will_be_read} we get at least one event
$\revent_{\event}$ in $\run$ that reads the value of $\inopof\thread\breg\yvar-$.

%
%

We now look at the tuple $\statusof{\pqstate}\creg\breg\xvar{\pqstate'}$. From inductive hypothesis we get that the value of this tuple must be outputted later on on some thread that is still in a  $ \mkrelof{\porel \cup \rfrel}^* \revent$ path.

This clearly means that now the event $\event$ will be assigned a $\rfrel$ edge towards another event which belongs in a thread which has a  $ \mkrelof{\porel \cup \rfrel}^* \revent$ path. This means that also $\event$ will have a  $ \mkrelof{\porel \cup \rfrel}^* \revent$ path.

%
%
%
%

  \end{itemize}
 We now check for tuples of type $\statusof{\pqstate}\breg\creg\xvar{\pqstate'}$. We see that many rules can cause this tuple to be added, but none of them are triggered because of a change to  the set of threads that exist in some  $ \mkrelof{\porel \cup \rfrel}^* \revent$ path (from inductive hypothesis).

 It remains to show that the value currently in $\creg$ will be outputted later on in one of the threads that are still in a  $ \mkrelof{\porel \cup \rfrel}^* \revent$ path.
 From Lemma \ref{lem:marked_registers_will_be_read} we know that the value stored in $\breg$ will definitely be outputted at some point later in $\run$, but we also want to show that this will happen from an event $\event_r$ that is on a $\mkrelof{\porel \cup \rfrel}^* \revent$ path (and thus the thread of $\event_r$ also is).

 We check what rules could have caused the addition of $\tuple{\qstate, \creg}$. The only ones of those for which the inductive hypothesis does not hold immediately are $\bluecircled{2}$, $\bluecircled{5}$, and $\bluecircled{6}$.

 However even for those the analysis is relatively straightforward. We only show the case for Rule $\bluecircled{2}$.

 \begin{itemize}
 \item Rule $\bluecircled{2}$, and $\event = \outopof\thread\breg\yvar$.
   We see that for this rule to be triggered there must have existed tuple
 $\statusof{\pqstate}\creg\thread\xvar{\pqstate'}$. From inductive hypothesis the thread $\thread$ is one that occurs in some $\mkrelof{\porel \cup \rfrel}^* \revent$ path. Moreover, the rule is triggered only there exists an output edge on that same thread $\thread$, which means that it is sound to add the tuple $\statusof{\pqstate}\breg\creg\xvar{\pqstate'}$ because the value of $\breg$ will be outputted later on (i.e. immediately), on a thread as requested.
 \end{itemize}
 \end{proof}

 \begin{corollary}
Let $\statusof{\pqstate}-\creg\xvar{\pqstate'}$ be a tuple that is created when running the rules of Figure \ref{fig:rules_wra}.
 Then,
 there exists path in $\rmachine$ starting in $\qstate$, and ending in some $ \revent$, such that $\revent$ is outputting the value of $\creg$.

 \end{corollary}
 For the reverse (which is not as complicated), we state:
\begin{lemma}\label{lem:path_implies_tuple}
Let $\run$ be the run of Proposition 
\ref{prop:splitting_violating_run}, and $\qstate\movesto{
\op}\pqstate$ be an input or output edge that 
is accessed during the $\run_{\beta}\cdot \wevent' \cdot\run_{\alpha}$ 
part of $\run$.
 Then for the 
 execution graph event $\event = \event(\op\sslash\val )= \eventof\ttype\thread\xvar\val$,
  corresponding to $\op$,
   we have that  $\event  \mkrelof{\porel \cup \rfrel}^* \revent$ implies $\statusof{\pqstate}-\creg\xvar{\pqstate'}$, at the execution graph 
   that has been produced by  $\mkegraph$.

\end{lemma}

\begin{proof}
We prove the claim by induction on the events of $\run_{\alpha}$, starting from the last one.

\underline{Base case: $\run_{\alpha} = \emptyset$}. In this case $\qstate$ is the only state accessed during the $\run_{\alpha}$ and the claim holds.

\underline{Inductive hypothesis:} We assume the claim holds for up to $n$ states reached backwards in $\run_{\alpha}$.

\underline{Inductive step:} We now take one more step backwards in $\run$. We will denote this step as  $\tuple{\qstate_{\alpha},\op,\pstate_{\alpha}}$, and proceed with a case analysis of the input operation $\op$, and the type of tuple associated with $\pstate_{\alpha}$ from the inductive hypothesis.


 	Assume that $\event  \mkrelof{\porel \cup \rfrel}^* \revent$ for $\event$ the event in the execution graph synchronizing with $\op$.
 	Since there exists such a path we take a step in this path.
 	This will be either through a $\porel$ or a $\rfrel$ edge connecting to an event $\pevent$, either on the same thread as $\event$, or reading its value.
 	From inductive hypothesis, the state where $\pevent$ originates is marked with some tuple in some  $\statusof{\pqstate}---{\pqstate'}$.
	We now need to guarantee that this tuple with be propagated backwards properly through the events of $\run$ 	between $\event$ and $\pevent$ to reach the point where it is assigned to $\qstate_{\alpha}$.
	We know that when the preceding events are copy edges then this holds.
	With a thorough inspection of the rules of Figure \ref{fig:rules_wra}, particularly for the transparency rules, we see that the only way the state of a tuple does not propagate to the previous state of an event is when there is input on the variable $\xvar$ of $\revent$, on the register where the current value of this event is.
	This is impossible due to Corollary \ref{cor:run_does_not_overwrite}.
	Thus if a  $\event  \mkrelof{\porel \cup \rfrel}^* \revent$ path exists, then some tuple will be associated to $\qstate_{\alpha}$.
\end{proof}

\begin{corollary}\label{cor:tuples_equal_runs}
Let $\run$ be the run of Proposition \ref{prop:splitting_violating_run}, and $\qstate\movesto{
\op}\pqstate$ be an edge that is accessed during the $\run_{\beta}\cdot \wevent' \cdot\run_{\alpha}$ part of $\run$, and synchronizes with event $\event$.
 Then,
 \begin{itemize}
 \item $\statusof{\pqstate}\thread\creg\xvar{\pqstate'}$ is added iff all events that take place in $\thread$ before $\event$ (including $\event$ have a $\mkrelof{\porel \cup \rfrel}^* \revent$ path, and
 \item $\statusof{\pqstate}\breg\creg\xvar{\pqstate'}$ is added to iff any event that writes the value that is to be in register $\breg$ upon reaching $\qstate$ will have a $\mkrelof{\porel \cup \rfrel}^* \revent$ path.
\end{itemize}
\end{corollary}
%
%
%

We are now in a very solid state regarding at least part of the meaning on the tuples of our algorithm.
We have established that the tuples correctly mark threads that have a $\mkrelof{\porel \cup \rfrel}^* \revent$ path, or will acquire one once they become the origin of a $\rfrel$ edge that targets a thread with such a path.
We break the remaining argument into the following questions:

\begin{enumerate}
    \item Do all threads that have such a path get assigned to a tuple with that thread?
    \item What does the transition to $\hidden$ tuples guarantee?
    \item Do our rules mark with $\hidden$ all the states and threads we want?
\end{enumerate}

To answer Question $1$, we have:
We just augured that when exploring a new event while traversing a violating run backwards, we will annotate the appropriate threads for all $\porel$ and $\rfrel$ edges that are created. We have also shown that if a tuple exists then so does a path. Therefore, as far as $\mkrelof{\porel \cup \rfrel}^* \revent$ path we have complete correspondence between paths and tuples.

As Proposition \ref{prop:splitting_violating_run} requests, we need to encounter a $\wevent'$ event, and after this we have the stronger condition of $\revent \mkrelof{\invmkrelof\rfrel} \wevent \mkrelof{\porel \cup \rfrel }\wevent'$
 to keep track of. However, we have already shown that most of the rules of Figure \ref{fig:rules_wra} correspond to $\mkrelof{\porel \cup \rfrel}^* \revent$ paths. This is where the rules $\largebluecircled{4}$ and $\largebluecircled{5}$ come into play. They are there to ensure that once encountering a $\wevent'$ event, then we will switch to $\hidden$ tuples.

Since we have that the algorithm of Figure \ref{fig:ghost_reads} has been proven correct, and has been pre-run before we run Algorithm \ref{algo:wra_pseudo}, all output events output something that has indeed been inputted at some point of the run.
Thus we are guaranteed that eventually, there will be events $\wevent$ and $\wevent'$ from Proposition \ref{prop:splitting_violating_run}.

It remains to show that out algorithm can detect the existence of the intermediate $\wevent'$ events, and that all the $\failconf$ configurations our algorithm marks would correspond to violations.
The first step towards this is to guarantee that when reading a  $\wevent'$ event, we will indeed start marking the states with $\hidden$ tuples.
We know that $\wevent'$ must occur, it must be on variable $\xvar$, and it must not be on the register that is currently the register that will be outputted on $\revent$.

Since the claim is that $\hidden$ tuples will be associated with the $\run_{\beta}$ part of the run, we see how we can guarantee that beginning to mark with $\hidden$ tuples must be because of such an event.
\begin{proposition}
\label{prop:initialize_beta}
Rules $\bluecircled{4}$, $\bluecircled{5}$ 
are only applied when we can guarantee the existence of event $\wevent'$ from Proposition \ref{prop:splitting_violating_run}.
\end{proposition}
\begin{proof}

From the rules of Figure \ref{fig:rules_wra}, we see to initiate the $\hidden$ tuples we must be using rules $\bluecircled{4}$, or $\bluecircled{5}$
.
We check what such rules mean. In the case of $\bluecircled{4}$ it is very clear that an event $\wevent'$ has occurred, and moreover, any event taking place on the same thread as $\wevent'$ will have a $\mkrelof{\porel \cup \rfrel}^*$ path to it.

Regarding rule $\bluecircled{5}$ we do not yet have the existence of such an event. However, since the rule is triggered by having output on $\xvar$ on a marked thread, we are guaranteed by the correctness of the algorithm of  Figure \ref{fig:ghost_reads} that there will exist a corresponding input for any possible run of $\rmachine$ leading to this output.
Thus, when such an event occurs, no matter in what thread it takes place, in $\run$, it will be a $\wevent'$ that has a $\mkrelof{\porel \cup \rfrel}^* \revent$ to $\revent$, while the event $\wevent$ guaranteed by Proposition \ref{prop:splitting_violating_run} will have a $\corel$ edge to it from the semantics of  Figure \ref{add:event:fig}.

 This means that if those write events (which have not occurred yet, but will while we traverse $\run$ backwards) took place in the thread marked by the rule, then Proposition \ref{prop:splitting_violating_run} would be satisfied.
 Therefore it is sound to mark this thread in $\hidden$, since all events that will occur in it have the required paths to $\wevent'$.
\end{proof}

We continue with describing the invariant of annotating $\beta$ tuples.

\begin{lemma}
Let $\run$ be the run of Proposition \ref{prop:splitting_violating_run}, and $\qstate\movesto{
\opof\thread---}\pqstate$ be an input or output edge that is accessed during the $\run_{\beta}\cdot \wevent'$ part of $\run$.
 Then the following hold:

 \begin{itemize}
 \item $\hiddenof{\qstate}\thread\creg\xvar{\pqstate}$ is added if $\thread$ occurs in some existing $ \mkrelof{\porel \cup \rfrel}^* \wevent'$  path or
 if $\op$ is an input on a new thread, but the value that is inputted has been later on read by an event in one of the threads in one $ \mkrelof{\porel \cup \rfrel }^* \wevent'$ path.
 \item $\hiddenof{\qstate}\breg\creg\xvar{\pqstate}$ is added if the set of threads that can reach the event $\wevent'$ though a \\
  $ \mkrelof{\porel \cup \rfrel}^* \wevent'$ path does not chance from $\pstate$ to $\qstate$, and the value currently in $\breg$ will be outputted later on in one of the threads that are in a  $ \mkrelof{\porel \cup \rfrel \cup}^* \wevent'$ path.
 \end{itemize}

\end{lemma}
The proof of this is entirely identical to that of Lemma \ref{lem:tuples_imply_paths} (since we apply the same rules).
The main thing to prove is the reverse, i.e.
\begin{lemma}\label{lem:beta_rel_associated_to_run}
Let $\run$ be the run of Proposition \ref{prop:splitting_violating_run}, and $\qstate\movesto{
\opof\thread---}\pqstate$ be an input or output edge that is accessed during the $\run_{\beta}\cdot \wevent'$ part of $\run$.
 Then for the event $\event$ synchronising with $\op$ in the relevant execution graph we have $\mkrelof{\invmkrelof\rfrel} \wevent \mkrelof{\porel \cup \rfrel }\wevent'$ implies $\hiddenof{\qstate}\thread\creg\xvar{\pqstate}$ is added.
 \end{lemma}

\begin{proof}
This proof is also very similar to that of Lemma \ref{lem:path_implies_tuple}, but in this case we do have some modifications. We use induction of the length of $\run_{\beta}$. The base cases are slightly altered as we already discussed in Proposition \ref{prop:initialize_beta}
The inductive hypothesis and and remaining issues are based exactly on marking new threads that become aware of existing marked tuples.
\end{proof}

Now it remains to argue that upon encountering the $\wevent$ we\textbf{ have, and can detect}, a violation.
We have already proved that for the execution graph $\egraph$ corresponding to $\run$, we have that our $\hidden$ tuples keep track of threads that have a $\event  \mkrelof{\porel \cup \rfrel \cup \corel_x}^* \wevent'$ path, for the appropriate $\wevent$.
We need to see that we can now detect all events that would be $\wevent$, when those occur.

\begin{lemma}\label{lem:failures_are_correct}
$\failconf\oplus\qstate$ iff there exists run as stated in Proposition \ref{prop:splitting_violating_run}.
\end{lemma}

\begin{proof}
($\Rightarrow$)

We see from the rules of Figure \ref{fig:rules_wra} that in order to mark an $\failconf$ state we need a $\hidden$ tuple, together with some event a $\opof\thread\xvar\areg-$ on a marked thread, in order to trigger rule $\bluecircled{7}$.

We can easily then check the semantics of execution graphs and confirm that this $\wevent$ will create a cycle.
This can be easily verified from the semantics in Figure \ref{add:event:fig}.
Thus marking an $\failconf$ configuration corresponds to a violation of $\rasem$.
For the reverse, we need to argue we capture \textit{all} violations.

 ($\Leftarrow$)

Assuming there exists such a violating run we are guaranteed to at some point reach the $\wevent$.
If this $\wevent$ occurs on some marked thread then we know we will mark a violation, as it will trigger rule $\bluecircled{7}$.

We now check if it is possible for $\wevent$ to occur in an unmarked thread.
Of course, we are still guaranteed that  Proposition \ref{prop:splitting_violating_run} holds, and therefore there exists a path to $\wevent'$.
We continue here with induction on the length of the path between $\wevent$ and $\wevent'$.

\underline{Base case: Path length = 1}. In this case it must be the case where immediately a cycle is created by adding the $\corel$ edges dictated by the semantics. Since $\wevent$ must be on an unmarked thread it cannot be on the thread of $\wevent'$ and thus they are not connected with a $\porel$ edge.
Moreover, they are both write events, and thus they are not connected by a $\rfrel$ edge either.

Therefore, it must be the case that $\wevent \corel_x \wevent'$, which can only happen if \textit{on the same run $\run$)}, we got to output the value of $\wevent'$, \textit{event though $\wevent$ was also $\visible$}.

This brings us to the situation where either on the same $\mkrelof{\porel \cup \rfrel}^*$ path three output events altered between reading from $\wevent$ and $\wevent'$, or
on two distinct $\mkrelof{\porel \cup \rfrel}^*$ paths had been formed in $\run$, and on those the two events were outputted in reverse order. In this case one of the versions of rule $\bluecircled{6}$ would have been applied earlier on when such outputs occurred.

By the correctness of the propagations that we proved earlier in Lemma \ref{lem:path_implies_tuple}, we have that a tuple $\hiddenof{\qstate}\areg\areg\xvar{\pqstate}$ would be available by the time we reach the $\wevent$, and this would mark the $\failconf$.

\underline{Hypothesis: all $\mkrelof{\porel \cup \rfrel}^*$ paths up to $n$ length are marked.}

\underline{Step: We take one more step backwards.}

We check the $\mkrelof{\porel \cup \rfrel}^*$ path that must connect us to $\wevent'$.
Since we need to be in an unmarked thread, this path cannot start from a $\porel$ edge (due to hypothesis).
Similarly, if it starts from a $\rfrel$ edge we fall in a similar scenario as the event that the $\rfrel$ is targeted cannot be marked.

Therefore the path must start from a $\corel$ edge, in which case we perform an analysis similar to the base case we did and we are done.
\end{proof}
The above establishes the correctness of Algorithm \ref{algo:wra_pseudo}. 

\section{Upper and lower bounds for $\raconsprob$ and $\sraconsprob$}\label{sec:hardness}

As we discussed above we were not able to provide an equally
efficient algorithm for $\raconsprob$ and $\sraconsprob$.
In this section we focus on the complexity of these two problems, and show that they turn out to be \PSPACE-complete.
We first show that both $\raconsprob$ and $\sraconsprob$ are \PSPACE-hard.
 Then we provide a \PSPACE~upper bound.
\subsection{The \PSPACE~lower bound}
We reduce from the problem of emptiness of DFA intersection. This problem is known to be \PSPACE-hard \cite{DFA_intersect_PSPACE_kozen}. The problem is defined as follows:

\begin{definition}[DFA intersection non-emptiness]
Given $k \in \mathbb{N}$ and $k$ DFAs $D_1, D_2, \ldots, D_k$ of at most $n$ states each, over the alphabet $\Sigma$, \textit{DFA intersection non-emptiness} is satisfied iff there exists a word $w \in \Sigma^*$ such that $w \in L(D_1) \cap L(D_2) \cap \ldots \cap L(D_k)$.
\end{definition}

Note that here, $k$ is part of the input (otherwise the problem can be solved in polynomial time), so the number of DFAs varies between the different instances of the problem.

In our reduction, we take the standard definition of DFAs~\cite[Definition~1.5]{Sipser:1102988}: a DFA is a tuple $D = (Q, q_0, \Sigma, \delta, F)$, where $Q$ is the set of \emph{states}, $q_0 \in Q$ is the \emph{initial state}, $\Sigma$ is the \emph{input alphabet}, $\delta : Q \times \Sigma \rightarrow Q$ is the \emph{transition function} and $F \subseteq Q$ is the set of \emph{final states}. The notions of run and accepted language are defined as usual. For a DFA $D$, we denote $L(D)$ its accepted language.

\begin{theorem}
$\raconsprob$ and $\sraconsprob$ are \PSPACE-hard.
\end{theorem}

The proof of this theorem mimics the technique used in \cite{RA_verification_2025} for the \NP-hard and \coNP-hard results. Namely, given the input to the \textit{DFA intersection non-emptiness} problem, we construct a register machine $\rmachine$ which has a cycle which violates $\rasem$ iff there exists an appropriate word $w$ in the intersection of the automata.

\begin{proof}
First, up to adding a termination symbol $\# \notin \Sigma$, we assume that for each $i$, the input DFA $D_i$ has a unique final state, that we denote $q_{i,f}$. Indeed, $w \in L(D_1) \cap L(D_2) \cap \ldots \cap L(D_k)$ if and only if $w\# \in L(D_1)\# \cap L(D_2)\# \cap \ldots \cap L(D_k)\#$.

So, for each $i$, we let $D_i = (Q_i, q_{i,0}, \Sigma, \delta_i, \{q_{i, f}\})$ (the alphabet $\Sigma$ is common to all automata). We also number states in each $Q_i$: $Q_i = \{ q_{i,j} \mid 0 \leq j \leq {|Q_i|}\}$.

The reduction consists in constructing in polynomial time a register machine $\rmachine$ whose size is polynomial in the size of the input (the $k$ automata).

Given the given the $k$ many DFAs, we parse them and create a register machine over:
\begin{enumerate}
\item The threads $\thread_{i}$ with one thread per automaton $D_i$
\item A single variable $\xvar$, and
\item Registers $r_{i,j}$ one for each state $q_{i,j}$ of each
 input automaton $D_i$.
 \textit{Intuitively, in the reduction that follows,
 a ``non-nil'' value in a register $r_{i,j}$ models
  that the automaton $D_i$, which we are simulating,
  ``is in state'' $q_{i,j}$}.
  We will be writing $\areg_{i}$ and $\breg_{i}$
  to denote the registers which correspond respectively to
   the initial and final states of each automaton.
    We make this distinction in notation because these are the
    only registers which will be involved in observable read and
    write events later on. The other ones will only be involved
     in copy events which are not shown in execution graphs.
     The register machine $\rmachine$ will also be handling an
     extra register called $r_{nul}$, which is independent from the input.
\end{enumerate}

The construction builds the register machine in three phases, the initialization phase, the ``shifting'' phase, and the output phase. They consist in the following:
\begin{itemize}
  \item \textbf{The initialization} of threads phase,
   where the register machine performs $2$ write events per
    automaton $D_{i}$ on the corresponing thread $\thread_i$, linked in a
     single path. Each thread $\thread_i$
      will first get a write event targeting
       the register $\breg_{i-1}$,
        followed by a write event targeting the register
         $\areg_{i}$, as shown below:
        \[
            q_{2i} \xrightarrow{\wopof{\thread_i}\xvar{\breg_{i-1}}} q_{2i+1} \xrightarrow{\wopof{\thread_i}\xvar{\areg_{i}}} q_{2i+2} ~.
        \]

        This part of the machine starts in state $q_0$, and ends in state $q_{2n-1}$. When counting the previous thread's registers we always use operations $mod~k$.

    \begin{figure}
        \centering
            \begin{tikzpicture}
             \node (init) at (4,3) {$\ldots$};

             \node (w1) at (6,4) {${\wopof{\thread_i}\xvar{\breg_{i-1}}} $};
             \node (w2) at (6,2) {${\wopof{\thread_i}\xvar{\areg_{i}}}$};
             \node (w3) at (9,4) {${\wopof{\thread_{i+1}}\xvar{\breg_{i}}} $};
             \node (w4) at (9,2) {${\wopof{\thread_{i+1}}\xvar{\areg_{i+1}}}$};
             \node (space) at (10.8,3) {$\ldots$};

             \node (w5) at (12.5,4) {${\wopof{\thread_{k}}\xvar{\breg_{k-1}}} $};
             \node (w6) at (12.5,2) {${\wopof{\thread_{k}}\xvar{\areg_{k}}} $};
              \node (wnul) at (16,4) {${\wopof{\thread_{nul}}\xvar{r_{nul}}} $};

             \draw[-stealth, \pocolor] (w1) -- (w2) node[midway, fill=white] {$\porel$};
             \draw[-stealth, \pocolor] (w3) -- (w4) node[midway, fill=white] {$\porel$};
            \draw[-stealth, \pocolor] (w5) -- (w6) node[midway, fill=white] {$\porel$};

         \end{tikzpicture}
        \caption{The execution graph at the end of the initialization phase}
        \label{fig:conp_init}
\end{figure}

Before this phase is over we also construct an extra thread called $\thread_{nul}$ and add a single write event on it on the special register $r_{nul}$. This write event is the only one that will ever take place on $\thread_{nul}$.

 When this part of the register machine $\rmachine$ is executed it will always create the same execution graph as depicted in Figure \ref{fig:conp_init}. 
 \item \textbf{The shifting phase,} in which we construct first the following \textbf{states}: \begin{itemize}
 \item A special \textbf{state} called $q_{main}$ which is immediately accessed after the initialization phase is over, with an $\epsilon$ transition.
 \item A set of \textbf{states} $Q_{\Sigma} = \{q_a \mid a \in \Sigma\}$, which are reached from $q_{main}$ with an $\epsilon$ transition. These will be used to non-deterministically select the next symbol of the word we are constructing in the search for a non-empty intersection.
 \item For each state of $Q_{\Sigma}$ a set of \textbf{states} $Q_{\Sigma,i} = \{ q_{a,i} \mid  0 \leq i \leq k\}$. These will be used to initiate the selection of a transition for the next automaton $D_i$, upon transition $a$.

 \item For each state of $Q_{\Sigma,i}$ we also create a set of \textbf{states} $Q_{\Sigma,i,j} = \{ q_{a,i,j} \mid 0 \leq j \leq |Q_i| \}$.

  These will be used to simulate a non-deterministically selected state which has an enabled transition with symbol $a$ by each automaton $D_i$. Since the automata are deterministic, there exists only one next state.
 \item finally, we create a set of states
  $Q_{aux} = \{q_{a,i,j, \ell} \mid a \in \Sigma,~ q_{i,j} \xrightarrow {a} q_{i,\ell}, \text{ and } 0\leq \ell \leq |Q_i| \}$,
  for each state $q_{i,j}$ which could take an $a$-transition.
   We create one such
  auxiliary state for each state of the automaton the
   state $q_{i,j}$ belongs in.
  This auxiliary set of states will be used later on to ensure
   that our ``simulation'' of a transition in an automaton does
    not create words the automaton would not actually produce.
 \end{itemize}

 Beyond the transitions we have already described, we also add the following \textbf{transitions:}
 \begin{itemize}
 \item for each $a \in \Sigma$, we add: $q_a \xrightarrow{\epsilon} q_{a,1}$.
 \item for each symbol $a \in \Sigma$, automaton $D_i$,
 and state $q_{i,j}$ in $D_i$ such that $\delta_i(q_{i,j},a) = q_{i,j'}$,
  we add the transitions:

 $$ q_{a,i} \xrightarrow{\ceventof{r_{i,j'}}{r_{i,j}}} q_{a,i,j}~ , \text{ and, }$$

we add transitions using the auxiliary states
 $q_{a,i,j,\ell}$ to copy the value of $r_{nul}$ into
 all registers $r_{i,\ell}$, with $\ell \neq j'$.
 From the last such state (called $q_{a,i,j,\ell_f}$
 for as much clarity as possible) we also add the transition:

$$ q_{a,i,j,\ell_f} \xrightarrow{\epsilon} q_{a,i+1}~ (mod ~k).$$

 \item For the states $q_{a,0}$, which are reached only when the last automaton has selected a transition to follow with the symbol $a$, we add transitions

 $$q_{a,0} \xrightarrow{\epsilon} q_{main}~.$$
\end{itemize}
 This completes the assignment phase. We add a transition $q_{main}
\xrightarrow{\epsilon} q_0^{out}$, which initiates the output phase.
  \item In the \textbf{output phase}, the register machine performs $n$ output transitions, one for each thread $\thread_i$

        \[
            q_{i}^{out} \xrightarrow{\outopof{\theta_{i}}{\xvar}{\areg_{i}}} q_{i+1}^{out}
        \]

\end{itemize}
\begin{figure}[t]
  \begin{center}
    \begin{tikzpicture}
        \node (w1) at (6,4) {$\eventof\wtype{\theta_{i}}{b_{{i-1}}}{x}$};
        \node (w2) at (6,2) {$\eventof\wtype{\theta_{i}}{a_{i}}{x}$};
        \node (r)  at (6,0)  {$\eventof\rtype{\theta_{i}}{a_{i}}{x}$};
        \node (w3) at (11,4) {$\eventof\wtype{\theta_{{i+1}}}{b_{{i}}}{x}$};

        \draw[-stealth, \pocolor] (w1) -- (w2) node[midway, fill=white] {$\porel$};
        \draw[-stealth, \pocolor] (w2) -- (r) node[midway, fill=white] {$\porel$};
        \draw[-stealth]  (w3) -- (r) node[near start, fill=white] {$\rfrel$};
        \draw[-stealth, dashed, red]  (w2) -- (w3) node[near start, fill=white] {$\corel_x$};
        \draw[-stealth, dashed, red]  (w1) -- (w3) node[near start, fill=white] {$\corel_x$};

    \end{tikzpicture}
  \end{center}
\caption{for an assignment of variables that violates $c_i$ we get the above local execution graph.}
\label{fig:local_clause_violation}
\end{figure}

The claim is that if there exists a word $w \in \Sigma^*$, with $w \in L(D_1) \cap L(D_2) \cap \ldots \cap L(D_k)$ then there is an $\rasem$ violation of this register machine. The main idea here is that only if the values inputted at the $\breg_i$ registers are all shifted to be written on the $\areg_i$ registers at the time of reading, then each thread will be reading a value from the thread ``next'' to it, thus acquiring locally for each thread the execution graph described in Fig. \ref{fig:local_clause_violation}.
%
\end{proof}

\subsection{Membership in \PSPACE}\label{sec:pspace}

The key idea of our \PSPACE~algorithm is that if the register machine
$\rmachine$ can produce an arbitrarily long violating run $\run$,
with arbitrarily many events in the corresponding execution graph
$\mkegraphof{\run}$, then it must also be able to produce a shorter
run $\run_s$, which has at most length $2\times |\threadset|^{2}$, and
whose egraph contains the same cycle.
This is because if a cycle ``enters'' the same thread
multiple times we are able to short-circuit the egraph cycle
by jumping directly from the first entry event (say $\event_1$)
point to the event where the cycle exits that thread for the final time (say $\event_2$)
without performing the intermediate steps and instead following $\porel$
edges, as seen in Fig. \ref{fig:short_circuit_pspace}.

\begin{figure}
  \centering
  \begin{tikzpicture}
      \node[fill=lightgray] (thread) at (2,5) {$\thread$};
      \node (inv1) at (4,4.5) {};
      \node (e1) at (2,4) {$\event_1$};
      \node (inv1') at (0,3) {};
      \node (inv3) at (2,2.25) {\rotatebox{90}{$\cdots$}};
      \node (inv2') at (0,1.5) {};
      \node (e2) at (2,0.5) {$\event_2$};
      \node (inv2) at (4,0) {};


      \draw[-stealth, \pocolor] (e1) -- (inv3) node[midway, fill=white] {$\porel$};
      \draw[-stealth, \pocolor] (inv3) -- (e2) node[midway, fill=white] {$\porel$};
      \draw[-stealth,thick, dashed,red] (inv1) -- (e1);
      \draw[-stealth,thick, dashed,red] (e1) -- (inv1');
      \draw[-stealth,thick, dashed,red] (inv2') -- (e2);
      \draw[-stealth, thick,dashed,red] (e2) -- (inv2);
  \end{tikzpicture}
\caption{Short circuiting a large cycle (red, dashed) in the execution graph of a
run, by following the $\porel$ edges of $\thread$.}
\label{fig:short_circuit_pspace}
\end{figure}
Thus we know that a minimal cycle for each thread enters each thread at most once and exits each thread at most once.
We therefore know that if a cycle can be created in some egraph, the relevant events in such cycle are linear to the size of the register machine.
We have already shown that if the register machine
$\rmachine$ can produce an arbitrarily long violating run $\run$,
with arbitrarily many events in the corresponding execution graph $\mkegraphof{\run}$,
then it must also be able to produce a shorter
run $\run_s$, which has at most length $2\times |\threadset|^{2}$, and
whose egraph contains the same cycle.

The \PSPACE-algorithm thus works as follows:
\begin{itemize}
    \item We guess which threads will be involved in a cycle, and we guess which transitions in the register machine correspond to entering and leaving a thread.
    \item Some of the edges in the egraph directly connect these events either though $\porel$ or $\rfrel$. For the events that are connected via some $\corel$ we guess linearly many events for each such edge, so we can guarantee the register machine would produce this dependency.
    \item At this point we have $n^2$ many events relevant to the cycle.
    \item We guess an order among these events, which we claim is the order in which they will occur in the register machine.
    \item The remainder work is to determine whether the register machine can produce a run that subsumes these events in the right order, so that the correct $\rfrel$ is created.
    \item Whether the events are reachable from eachother is a check we can easily perform in polynomial time. What is not trivial is to track the locations on which the relevant values for the guessed events are stored. Indeed we can track these location in an array of at worst $n$ size, where each index can take as a value a different register of the register machine.
    \item We therefore start running the register machine non deterministically, and at each step we guess whether the next step is going to be one of the reserved events for the cycle, or just an event that will change the configuration of the register machine. (as per the rules of Fig. \ref{fig:semantics_exec_graphs}.
    \item Since we only have $n^2$ reserved events we know that these non-deterministic checks are polynomially many. However, the register machine can change exponentially many configurations (pairs of state and register contents). This means that we might possibly have to guess exponentially many times that the machine is changing configuration.
  \end{itemize}

 Throughout all this we are storing only the current register machine configuration as well as the the intersection of the execution graph that has been produced by the run with the initial cycle we have guessed as the possible violation.
   Thus the total space needed remains polynomial, but the execution time is exponential.


\subsection{Interpreting the hardness result and extension to PSI and SC}

Our reduction above from the problem of $k$-DFA intersection non-emptiness
to the problem of $\raconsprob$ showcases that any algorithm able to
input a register machine and decide whether all runs of such a machine
respect the $\rasem$ semantics, must require at least polynomial space.

Therefore, consistency check for
\textit{any memory model stronger than $\rasem$} over register
machines is also \PSPACE-hard.

This imediately implies the \PSPACE hardness of $\sraconsprob$, which
is also in \PSPACE{} as shown in Section \ref{sec:pspace}.
However there are other declarative memory models where the definition is
strictly stronger than $\rasem$.
We highlight for example the
 declarative definition of PSI (\cite[Definition 4]{RaadLV18_PSI}, \cite{SovranPAL11_PSI}),
  which requires the acyclicity of
$\corel\cdot\tclosureof{\mkrelof{\porel\cup\rfrel\cup\partcorel}}\cdot \invmkrelof\rfrel$, i.e. only slightly
more restrictive from the definition
of $\srasem$ (Definition \ref{def:egraph_satisfies_memory_model}).

Finaly, the memory model of SC is also defined declaratively
  (as an acyclicity condition of
   $\tclosureof{\mkrelof{\porel\cup\rfrel\cup\partcorel\cup\frrel}}$~\cite{DBLP:journals/toplas/AlglaveMT14}), and is also stricter than $\rasem$,
   which means its consistency checking over register machines inherits our
   lower bound.
   This improves the previously known lower bound, which was a trivial
   \NP-hard lower bound, inherited
   from the relevant testing problem \cite{DBLP:journals/siamcomp/GibbonsK97}.
The respective upper bounds for either PSI or CS are not currently known,
and in the case of SC it would not be suprising if the
problem is undecidable, as it is in the case of data-dependent model \cite{DBLP:journals/iandc/AlurMP00}.
We leave these two questions as future endavours, as we believe our algorithms
will be needing significant improvements to be able to handle such cases.


%

\section{Conclusion and Future Work}\label{sec:conclusions}
In this paper, we have taken the first steps towards a framework for  verification under the RA semantics and its variants.
To that end, we prove a polynomial space upper bound (and polynomial time in the case of $\wraconsprob$) when the implementation is described in the classical register machine model.

Our first future endeavor is to improve and generalise our prototype tool implementation
 \cite{codelink} of our algorithms. 
  The tool suite can also be used to showcase how to model different popular
  distributed systems paradigms such as vector clocks, buffers, broadcast comands etc,
  to establish the generality of a register machine as a model.
We also plan to
leverage
abstraction and stateless model-checking techniques to achieve more efficient verification.

%

%
Aside from the above, we plan to consider more expressive modeling languages that make our work applicable to a wider class of protocols at the cache, compiler, and application levels.
One such extension would be to consider data-dependent register machines, which could model for example \texttt{compare-and-swap} events.
This suggests that one would have to study a theory of equality/inequality to create an augmented register machine space and characterize all the possible executions.
 In this case, we would expect a (much) higher complexity, whether we model this data dependency as register automata~\cite{DBLP:journals/tcs/KaminskiF94} or nominal automata~\cite{DBLP:journals/corr/BojanczykKL14}, or in the framework of well-structured systems \cite{DBLP:conf/lics/AbdullaCJT96,DBLP:journals/tcs/FinkelS01}.
Another exciting extension to our register machine formalism would be to allow for transitions encoding more complex actions, such as  \texttt{broadcast}, \texttt{rendezvous}, and \texttt{fences}.
We already know that these operations can be encoded with a series of transitions in the simple model, but having them explicitly as part of the syntax would create more succinct models and potentially speed up the verification.
Finally, another extension to the register machines would be in the direction of \textit{parameterized verification}. We would be particularly interested in enhancing the model with the ability to handle arbitrarily many threads interacting with the memory without having to hard-code them as part of the register machine.
%

%
Another direction of future work is to study the decidability
 and complexity of consistency checking for other memory models, such as $\psisem$, $\scsem$,
 and TSO.
A fundamental characteristic
 of our approach is using the declarative definition style. As discussed already
 $\scsem$ and $\psisem$ inherent our lower bounds
 due to their relation to the $\raconsprob$.
 However their upper bounds are not yet known.
 In the case of $\psisem$ it might be possible to extend our algorithmic approach
 without significant overhead.
 SC however, even though it is also defined declaratively
  (as an acyclicity condition of
   $\tclosureof{\mkrelof{\porel\cup\rfrel\cup\partcorel\cup\frrel}}$),
    is much harder to capture with our existing rules,
     as it uses a whole new type of relation,
     namely the $\mathsf{from-read}$ ($\frrel$ \cite{DBLP:journals/toplas/AlglaveMT14}) relation.


\bibliographystyle{plain}

\bibliography{references}

\end{document}